\documentclass[runningheads, 11pt]{llncs}
\usepackage[margin=1in]{geometry}
\usepackage[T1]{fontenc}
\usepackage{graphicx}
\usepackage{macros}

\begin{document}
\title{The Complexity of Justified Representation with Additive Utilities}
\author{Carmel Baharav\inst{1}\orcidID{0009-0004-3634-6721} \and
Jakob de Raaij\inst{2}\orcidID{0009-0003-5568-7751} \and
Agnès Totschnig\inst{1}\orcidID{0009-0004-7716-8183}}
\institute{Massachusetts Institute of Technology, Cambridge MA 02139, USA
\email{\{cbaharav,agnest\}@mit.edu}
\and
Harvard University, Cambridge MA 02138, USA
\email{jderaaij@fas.harvard.edu}\\}
\maketitle
\begin{abstract}
We study the computational complexity of satisfying proportional representation\emdash in particular proportional, extended, and fully justified representation (PJR, EJR, and FJR)\emdash in participatory budgeting and committee elections with additive utilities. First, we give a complete picture of the complexity of the axioms for a constant number of voters or voter types. Second, we show that even for committee elections with integer utilities bounded above by a small constant, satisfying FJR is intractable, giving the first strong \textbf{NP}-hardness result for a justified representation axiom. Third, we extend the Expanding Approvals Rule to committee elections with additive utilities and show that it satisfies PJR. Lastly, we show that no sequential voting rule can improve on the known positive result, thus proving that novel, substantially different voting rules are needed to surpass these boundaries. Beyond their theoretical merit, our results carry practical importance, as multi-winner voting with additive utilities has recently been gaining prominence in online deliberation and real-world participatory budgeting.

\keywords{Participatory Budgeting  \and Multi-Winner Voting \and Justified Representation }
\end{abstract}
\section{Introduction} 
There are numerous high-stakes democratic processes where the outcome is a subset of the alternatives, including committee elections, participatory budgeting, and deliberation. In recent years, these processes have become even more widespread, with participatory budgeting (PB) in particular exploding in popularity \cite{de2022international}. One key objective in selecting the outcome in these settings is \emph{proportional representation}: Each cohesive group of voters should be entitled to selecting a fraction of the outcome proportional to their size. The computational social choice community has responded to this with a commensurate flurry of work focused on formalizing this ideal and establishing a range of demanding proportionality axioms, known as \emph{justified representation}. 

Existing work on justified representation with additive utilities has primarily concentrated on two settings: The general participatory budgeting setting where voters have general, additive utilities for alternatives and alternatives have arbitrary costs \cite{rey2025computationalsocialchoiceindivisible}, and its special case where utilities are binary and all costs are equal, called committee elections with approval preferences \cite{lackner_overview}. 

Proportional representation was first formalized for committee elections with approval preferences in the seminal work of Aziz et al. \cite{aziz-jr-ejr}, in which the authors introduce a compelling notion of proportionality called \emph{Extended Justified Representation (EJR)}. Sánchez-Fernández et al. \cite{sanchez2017proportional} defined a slightly weaker axiom called \emph{Proportional Justified Representation (PJR)} that is implied by EJR. Several voting rules computable in polynomial-time and satisfying EJR \cite{aziz2017polynomialtimealgorithmachieveextended,peters2020proportionality} and PJR \cite{brill2024phragmen} are known.

On the other end of the spectrum, Los et al. \cite{los2022proportional} extend PJR and Peters et al. \cite{equalshares_fjr} extend EJR to the general participatory budgeting setting. Peters et al. also  revisit a polynomial-time algorithm called \textit{Method of Equal Shares (MES)} \cite{peters2020proportionality}, and show that it achieves EJR for binary utilities
and \emph{Extended Justified Representation up to 1 Alternative (EJR1)}, a relaxed version of EJR, in the general utility setting. 
They complement these positive results with the observation that \emph{even for just one voter}, finding an outcome satisfying EJR in the general utility and cost setting is weakly \textbf{NP}-hard\footnote{Throughout this paper, we say that a search problem is \textbf{NP}-hard if a polynomial-time algorithm would imply $\textbf{P}=\textbf{NP}$.}, by a reduction from the Knapsack problem.

Moreover, Peters et al. introduce a more stringent notion of justified representation called \emph{Fully Justified Representation (FJR)}, which implies EJR. They show that an outcome satisfying FJR always exists and can be found via their \emph{Greedy Cohesive Rule (GCR)}, which is \textbf{NP}-hard to compute. FJR is stronger than EJR even for committee elections with approval preferences; whether a polynomial time computable voting rule satisfying FJR exists is an important open problem \cite[Q5]{lackner_overview}.

Here, we set out to advance the broader agenda by thoroughly understanding the computational hardness of satisfying the justified representation axioms for additive utilities. In particular, we focus on the intermediate cases between the two extremes of committee elections with approval preferences and the general participatory budgeting setting. Not only are they of independent interest, it also progresses us towards a complete picture of the computational tractability of the two more commonly studied extremes.

Our first research question is motivated by the striking difference between the two settings for a single voter, and the fact that Knapsack admits pseudo-polynomial time algorithms: \begin{center}
    \textbf{Question 1.} For a constant number of voters or voter types, what precise utility and cost settings cause the shift in computational tractability of finding outcomes satisfying PJR/EJR/FJR?
\end{center}
We see the value of this question mostly in identifying the ``hardness boundary'' of the justified representation axioms. Nonetheless, there exist some practical scenarios with a very small number of voters (or voter types), such as a city seeking to be  proportionally representative to its districts when selecting a set of measures to take, each of which will benefit some districts more than others.
 
For a general number of voters and an arbitrarily structured voter pool, there is also a distinct gap in the existing complexity landscape for the setting of additive utilities and unit costs: While we know that MES gives EJR for approval utilities and arbitrary cost, it is only guaranteed to give an EJR1 outcome if utilities are arbitrary, even for unit costs. However, unlike in the general setting, there is no known intractability obstacle to satisfying PJR, EJR, and FJR with unit costs.
 This motivates our second research question: \begin{center}
    \textbf{Question 2.} Where is the ``hardness boundary'' for a general number of voters in the additive utilities, unit costs setting?
\end{center}

Besides the theoretical value of fully understanding how each of the parameters affects the computational tractability, this intermediate setting is also interesting from a practical perspective. In AI-aided deliberation, Fish et al. \cite{fish2025generativesocialchoice} and De et al. \cite{de2025questionquestionsauditingrepresentation} use general additive utilities to find a slate of statements satisfying justified representation to proportionally summarize the opinions of the voters.
In participatory budgeting elections, numerous cities like Strasbourg and Gdansk \cite{yang2024designing} allow voters to express additive utilities beyond approval voting, while some cities, like New York City \cite{NYC}, allocate a fixed amount of money to every proposed project, essentially enforcing unit-cost. To our knowledge, no city currently does both concurrently (soliciting additive utilities but fixing the costs of projects to be the same), but it is conceivable that cities could want to do this --- it enables voters to have more expressive preferences without having to do complex reasoning or strategizing about budget constraints. 

 While eliciting exact cardinal utilities of the voters is challenging, various mechanisms have been proposed: In online deliberation, Fish et al. \cite{fish2025generativesocialchoice} use LLMs to estimate the cardinal utility of a voter for a statement, given survey responses by this voter. Ideas For Change\footnote{\url{https://www.ideasforchange.org}}, an online deliberation platform deployed in various cities across the US, allows users to hold a button longer to indicate stronger support, thus a larger utility. In participatory budgeting, elections have used variants of range voting, scoring rules, and cumulative voting for obtaining cardinal utilities \cite{WikiPB}. Moreover, cost utilities (i.e., assuming a voter's utility for a project is the project's cost if they approve it, zero otherwise) are natural, additive utilities that can be obtained from approval votes.

Finally, we note that almost all polynomial-time voting rules that exist in the multi-winner voting literature are \emph{sequential}: They add one alternative at a time to the outcome, deciding on the next alternative only based on the alternative itself and the outcome up to this point. At the same time, voting rules such as GCR that satisfy the strongest proportionality notions `plan ahead' and may select an alternative only if it complements well with other alternatives. This motivates our final research question:
\begin{center}
    \textbf{Question 3.} Are there settings in which no sequential voting rules or voting rules with bounded ``planning ahead'' can satisfy PJR/EJR/FJR?
\end{center}

\subsection{Contributions}

 In \Cref{sec:constant-n}, we give a complete account of the complexity of finding outcomes satisfying PJR, EJR, and FJR for instances with \emph{a constant number of voters}, thus answering \textbf{Question 1}. First, we show that as long as all utilities are bounded by a polynomial\footnote{Whenever we state that utilities or costs are bounded by a polynomial, this polynomial is with respect to the size of the election instance, formally defined in \Cref{sec:model}. Equivalently, from a complexity viewpoint, this means that the bounded quantity is given in unary.}, GCR\emdash satisfying FJR\emdash can be implemented in polynomial time using dynamic programming (DP). We then design a modified version of GCR that satisfies EJR and, for costs bounded by a polynomial, can be implemented in polynomial time, again using DP. This only leaves the question of FJR with general utilities and polynomially-bounded costs; we show that in this setting finding an outcome satisfying FJR is \textbf{NP}-hard, even for unit costs and just two voters.
Interestingly, we observe that there is an asymmetry between the effect of general utilities and general costs: General utilities render FJR weakly \textbf{NP}-hard even for two voters, whereas the same does not hold in reverse. Our results for instances with a constant number of voters are summarized in \Cref{tab:constant-n-results}. Additionally, in \Cref{app:bounded-voter-types}, we show that our results for a constant number of voters still apply if the number of distinct voter \emph{types}, i.e. the number of distinct utility functions, is constant. 

\begin{table}[htbp]\vspace{-1.5em}
    \centering
    \caption{Complexity Results for a Constant Number of Voters} 
    \label{tab:constant-n-results}
    \medskip 
    
    \setlength{\tabcolsep}{4pt} 
    \renewcommand{\arraystretch}{1.4} 
    \begin{NiceTabular}{l c c c}
        \toprule
        \textbf{Utilities \textbackslash{} Costs} & \textbf{Unit} & \textbf{Polynomial} & \textbf{General} \\
        \midrule
        \textbf{Approval}   & \cellcolor{cyan!30} & \cellcolor{cyan!30} & \cellcolor{cyan!30} \\
        \textbf{Polynomial} & \cellcolor{cyan!30} & \cellcolor{cyan!30} & \cellcolor{cyan!30}\Cref{thm:fjr-constant-n} \\
        \textbf{General}    & \cellcolor{yellow!40}\Cref{thm:fjr-weak-hardness} & \cellcolor{yellow!40}\Cref{thm:mgcr-constant-n} & \cellcolor{red!40}\cite[Prop. 1]{equalshares_fjr}  \\
        \bottomrule
    \CodeAfter
    \tikz \draw[white, line width=0.5pt]
        ([yshift=-0.35pt]2-|3) -- (5-|3);
    \tikz \draw[white, line width=0.5pt]
        ([yshift=-0.35pt]2-|4) -- (5-|4);
    \tikz \draw[white, line width=0.5pt] (3-|2) -- (3-|5);
    \tikz \draw[white, line width=0.5pt] (4-|2) -- (4-|5);
    \end{NiceTabular}%
    \hspace{1.5em}%
    \renewcommand{\arraystretch}{1.2}
    \footnotesize 
    \begin{tabular}{@{}ll@{}} 
        \fcolorbox{gray}{cyan!30}{\makebox[1em]{\rule{0pt}{1ex}}} & $\P$ for FJR, EJR, PJR \\
        \fcolorbox{gray}{yellow!40}{\makebox[1em]{\rule{0pt}{1ex}}} & $\P$ for EJR \& PJR, hard for FJR \\
        \fcolorbox{gray}{red!40}{\makebox[1em]{\rule{0pt}{1ex}}} & \textbf{NP}-hard for PJR, EJR, and FJR \\
    \end{tabular}
\end{table}

For unbounded numbers of voters, we show in \Cref{sec:fjr-hardness} that finding an outcome satisfying FJR is strongly \textbf{NP}-hard, even for unit-costs and utilities in $\set{0,1,...,5}$. To the best of our knowledge, this is the first proof of \emph{strong} \textbf{NP}-hardness for a justified representation axiom, and in fact also the only hardness result beyond the one-voter Knapsack reduction due to \cite{equalshares_fjr}. We furthermore show that the proof that satisfying FJR is strongly \textbf{NP}-hard also holds in the setting of cost-utilities in $\set{0,1,...,5}$.

In \Cref{sec:general-n-pjr}, we extend a polynomial-time voting rule called the \emph{Expanding Approvals Rule} defined for ordinal preferences or clustering to the setting of committee elections with additive utilities. We prove that it satisfies PJR. Thus, we make significant progress towards answering \textbf{Question 2}: While neither negative nor positive results about proportionality for committee elections with additive utilities were known before, our results imply that PJR is achievable by a reasonable voting rule while FJR is hard. 
Our results are summarized in \Cref{tab:general-n-results}.  Interestingly, and in contrast to our results for a constant number of voters, the hardness boundary here seems to lie between approval utilities and utilities bounded above by an arbitrary constant as well as unit costs and costs bounded above by an arbitrary constant, not between polynomially-bounded and general utilities or costs.

\begin{table}[htbp]
    \centering
    \caption{Complexity Results for a General Number of Voters} 
    \label{tab:general-n-results}
    \medskip
    
    \begin{tabular}{@{} M{3in} @{\hspace{1.5em}} M{2.85in} @{}}
        
        \renewcommand{\arraystretch}{1.4} 
        \begin{NiceTabular}{l C{0.6in} C{0.6in} C{0.69in}}
            \textbf{PJR}  \\
            \toprule
            \textbf{Util \textbackslash{} Costs} & \textbf{Unit} & \textbf{Constant} & \textbf{General} \\
            \midrule
            \textbf{Approval}   & \cellcolor{cyan!30} & \cellcolor{cyan!30} & \cellcolor{cyan!30}  \\
            \textbf{Constant} & \cellcolor{cyan!30}& \cellcolor{orange!50} \Cref{thm:pjr-impossible-for-sequential-cost-util} & \cellcolor{orange!50} \\
            \textbf{General}    & \cellcolor{cyan!30}\Cref{thm:pjr-from-ear} & \cellcolor{orange!50}  & \cellcolor{red!40}\cite[Prop. 1]{equalshares_fjr}  \\
            \bottomrule
        \CodeAfter
            \tikz \draw[white, line width=0.5pt] ([yshift=-0.35pt]3-|3) -- (6-|3);
            \tikz \draw[white, line width=0.5pt] ([yshift=-0.35pt]3-|4) -- (6-|4);
            \tikz \draw[white, line width=0.5pt] (5-|2) -- (5-|5);
            \tikz \draw[white, line width=0.5pt] (4-|2) -- (4-|5);
        \end{NiceTabular}

        &

        \setlength{\tabcolsep}{2pt} 
        \renewcommand{\arraystretch}{1.4} 
        \begin{NiceTabular}{l C{0.6in} C{0.6in} C{0.69in}}
            \textbf{EJR}  \\
            \toprule
            \textbf{Util \textbackslash{} Costs} & \textbf{Unit} & \textbf{Constant} & \textbf{General} \\
            \midrule
            \textbf{Approval}   & \cellcolor{cyan!30} & \cellcolor{cyan!30} & \cellcolor{cyan!30}\cite[Thm. 2]{equalshares_fjr}  \\
            \textbf{Constant} & \cellcolor{orange!50} \Cref{thm:ejr-impossible-for-sequential-committee} & \cellcolor{orange!50} & \cellcolor{orange!50} \\
            \textbf{General}    & \cellcolor{orange!50} & \cellcolor{orange!50}  & \cellcolor{red!40}  \\
            \bottomrule
        \CodeAfter
            \tikz \draw[white, line width=0.5pt] ([yshift=-0.35pt]3-|3) -- (6-|3);
            \tikz \draw[white, line width=0.5pt] ([yshift=-0.35pt]3-|4) -- (6-|4);
            \tikz \draw[white, line width=0.5pt] (5-|2) -- (5-|5);
            \tikz \draw[white, line width=0.5pt] (4-|2) -- (4-|5);
        \end{NiceTabular}

        \\ \addlinespace[1.5em]
        
        \setlength{\tabcolsep}{2pt} 
        \renewcommand{\arraystretch}{1.4} 
        \begin{NiceTabular}{l C{0.6in} C{0.6in} C{0.69in}}
            \textbf{FJR}  \\
            \toprule
            \textbf{Util \textbackslash{} Costs} & \textbf{Unit} & \textbf{Constant} & \textbf{General} \\
            \midrule
            \textbf{Approval}   & \cellcolor{gray!30} & \cellcolor{gray!30} & \cellcolor{gray!30}\\
            \textbf{Constant} & \cellcolor{red!40}\Cref{thm:fjr-strong-hardness} & \cellcolor{red!40} & \cellcolor{red!40} \\
            \textbf{General}    & \cellcolor{red!40}& \cellcolor{red!40}  &  \cellcolor{red!40}\\
            \bottomrule
        \CodeAfter
            \tikz \draw[white, line width=0.5pt] ([yshift=-0.35pt]3-|3) -- (6-|3);
            \tikz \draw[white, line width=0.5pt] ([yshift=-0.35pt]3-|4) -- (6-|4);
            \tikz \draw[white, line width=0.5pt] (5-|2) -- (5-|5);
            \tikz \draw[white, line width=0.5pt] (4-|2) -- (4-|5);
        \end{NiceTabular}

        &
        
        \renewcommand{\arraystretch}{1.2} 
        \footnotesize 
        \begin{tabular}{@{}ll@{}} 
            \fcolorbox{gray}{cyan!30}{\makebox[1em]{\rule{0pt}{1ex}}} & \textbf{P} \\
            \fcolorbox{gray}{gray!30}{\makebox[1em]{\rule{0pt}{1ex}}} & unknown\\
            \fcolorbox{gray}{orange!50}{\makebox[1em]{\rule{0pt}{1ex}}} & unknown; no sequential voting rule\\
            \fcolorbox{gray}{red!40}{\makebox[1em]{\rule{0pt}{1ex}}} & \textbf{NP}-hard \\
        \end{tabular}\\
        
    \end{tabular}
\end{table}

In \Cref{sec:greedy-limitations} we show that the positive results in \Cref{tab:general-n-results} are as good as sequential voting rules can do, thus answering \textbf{Question 3}: No sequential voting rule satisfies EJR for unit-costs and utilities in $\set{0,1,2,3}$ and no sequential voting rule satisfies PJR for cost-utilities with costs in $\set{0,1,2,3}$. We show that our results extend to $O(1)$ utilities and cost-utilities, even when the sequential voting rules are allowed to take $O(1)$ alternatives outside of the outcome into account at once, a natural extension. Since almost all polynomial-time voting rules known to satisfy justified representation axioms beyond approval utilities are sequential, these limitations imply the need for novel voting rules, substantially differing from all existing voting rules.

\section{Preliminaries}\label{sec:model}

In this section we briefly introduce the setting of our theoretical results. We generally follow the notation introduced by Peters et al. in \cite{equalshares_fjr}. We use $\N=\set{0,1,2,...}$ and $\N^+ = \set{1,2,3,...}$. Furthermore, $[x]=\set{1,...,x}$.

\subsection{Combinatorial Participatory Budgeting}

A \emph{participatory budgeting (PB)} instance with additive utilities is a tuple $\elec=(N,C,B,\cost, \paran{u_i}_{i\in N})$. $N$ are the $n$ \emph{voters} and $C$ are the $m$ \emph{alternatives} (or candidates). Each alternative $c\in C$ has a \emph{cost}, given by $\cost \colon C \to \N^+$. For $T\subseteq C$, we write $\cost(T)=\sum_{c\in T} \cost(c)$ for the total cost of $T$. The total \emph{budget} available is $B \in \N^+$. 
Each voter $i\in N$ has an additive utility function specified by  $u_i\colon C \to \N$, so that the utility of voter $i$ for a set of alternatives  $T\subseteq C$ is $u_i(T)=\sum_{c\in T} u_i(c)$.\footnote{Instances with costs and $B$ in $\Q$, of course, can be turned into instances with costs in $\N^+$ by multiplying the costs and $B$ by the least common multiple of the denominators. The same holds true for utilities.} An \emph{outcome} is a set $W\subseteq C$, and voter $i$'s utility for that outcome is $u_i(W)$. An outcome is \emph{feasible} if $\cost(W) \leq B$. A \emph{voting rule} $\mathcal{R}$ is a function that takes in a PB instance $\elec$ and returns a feasible outcome $\mathcal{R}(\elec)$.

There are some notable special cases of this framework. If $\cost(c)=1$ for all $c\in C$, we say that the instance has \emph{unit-costs} and refer to it as a \emph{committee election} instance. This corresponds to a setting where a fixed maximum number of alternatives can be chosen for the outcome, the \emph{committee}. By convention, we use $k$ instead of $B$ for the budget and refer to it as the \emph{committee size}.

Independently, if $u_i(c)\in\set{0,1}$ for all voters $i\in N$ and alternatives $c\in C$, we say that the instance has \emph{approval preferences}, and if $u_i(c)\in\set{0,\cost(c)}$ for all voters $i\in N$ and alternatives $c\in C$, we say that the instance has \emph{cost utilities}. These correspond to a setting where voters approve of some subset of alternatives and disapprove of all other alternatives. A voter's utility for a set of alternatives $T\subseteq C$ then simply is, in the case of approval preferences, the number of alternatives in $T$ they approve, and in the case of cost utilities, the total cost of alternatives in $T$ they approve. 

\subsection{Justified Representation}

A commonly sought-after goal in participatory budgeting and committee elections is \emph{proportional representation}: If a sufficiently cohesive group of voters makes up $x$\% of the voting body, they should hold power over $x$\%, or more, of the budget. More formally, a coalition $S\subseteq N$ of voters can demand any set $T\subseteq C$ of alternatives for which $\frac{\abs{S}}{n}\geq \frac{\cost(T)}{B}$, i.e., the budget share taken up by $T$ is no larger than the population share of $S$. If the voters in $S$ cohesively prefer $T$ to the entire outcome $W$, proportionality is violated, since $S$ has disproportionally little utility from the outcome. The subtlety lies in formalizing when $S$ \emph{cohesively prefers} $T$ over $W$. We focus on three well-known variants:

    \begin{definition}\label{def:proportionality-axioms} 
Let $W \subseteq C$ be an outcome of a PB instance. The outcome $W$ satisfies:
\begin{alignat*}{2}
    &\bullet \text{ \textbf{Proportional Justified Representation (PJR)} \cite{sanchez2017proportional,los2022proportional} if} \quad & \sum_{c \in W} \max_{i \in S} u_i(c) &\geq \sum_{c \in T} \min_{i \in S} u_i(c), \\
    &\bullet \text{ \textbf{Extended Justified Representation (EJR)} \cite{aziz-jr-ejr,equalshares_fjr} if } \quad & \max_{i \in S} u_i(W) &\geq \sum_{c \in T} \min_{i \in S} u_i(c), \\
    &\bullet \text{ \textbf{Fully Justified Representation (FJR)} \cite{equalshares_fjr} if } \quad & \max_{i \in S} u_i(W) &\geq \min_{i \in S} u_i(T),
\end{alignat*} for all coalitions $S \subseteq N$ and alternative sets $T \subseteq C$ such that $\frac{\abs{S}}{n} \geq \frac{\cost(T)}{B}$.

A voting rule $\mathcal R$ satisfies PJR, EJR, FJR, if for any PB instance $\elec$, $\mathcal R(\elec)$ satisfies PJR, EJR, FJR, respectively.
\end{definition}

PJR, EJR, and FJR vary in how ``attached'' voters are to alternatives in $W$ and how ``flexible'' they are in agreeing on $T$. FJR is violated if the voter in $S$ least-satisfied with $T$ is still better off than the voter in $S$ most-satisfied with $W$. EJR removes the flexibility of the voters in supporting $T$: The joint utility of $S$ for alternatives $c\in T$ is determined by the minimum utility of a voter in $S$ for $c$. EJR is violated if the sum over those minimum utilities exceeds the largest utility of a voter in $S$ from $W$. Lastly, PJR makes voters more attached to alternatives in $W$: The joint utility of $S$ for $W$ is determined by the maximum utility of a voter in $S$ for each $c\in W$. PJR is violated if this joint (maximum) utility of $S$ for $W$ is less than their joint (minimum) utility for $T$. Thus, it is not hard to see that $(S,T)$ violating PJR also violate EJR, and $(S,T)$ violating EJR also violate FJR. Less obviously, it holds true that an outcome satisfying FJR (and thus EJR and PJR) always exists.

\begin{proposition}[\cite{equalshares_fjr,los2022proportional}]\label{prop:implications}
    If an outcome satisfies FJR, it satisfies EJR. If an outcome satisfies EJR, it satisfies PJR.
\end{proposition}

\begin{proposition}[\cite{equalshares_fjr}]\label{prop:fjr-existance}
    For any PB instance, there exists an outcome satisfying FJR.
\end{proposition}

There are two other natural choices for formalizing when $S$ \emph{cohesively prefers} $T$ over $W$ that we consider out of scope: The most natural choice, $u_i(W) < u_i(T)$ for all $i\in S$, leads to a proportionality axiom known as the \emph{core} \cite{aziz-jr-ejr}, which may be unsatisfiable \cite{fain2018fairallocationindivisiblepublic}, even in committee elections \cite{equalshares_fjr}.\footnote{It is unknown whether the core is always satisfiable for committee elections with approval preferences.} The fourth option of putting the sum inside or outside of the max/min in \Cref{def:proportionality-axioms}, $\sum_{c \in W} \max_{i \in S} u_i(c) \geq \min_{i \in S} u_i(T)$, leads to the natural generalization of a lesser-known proportionality axiom called \emph{FPJR} \cite{kalayci2025full} from approval-based committee elections to PB.

\section{Constant Number of Voters}\label{sec:constant-n}

As observed in \cite{equalshares_fjr}, the classic Knapsack problem can be reduced to finding an outcome satisfying EJR in the setting of general utilities and general costs \textit{for just a single voter}.
The same observation applies to PJR and FJR. 

\begin{remark}\label{rem:hardness-constant-n}
    Unless $\P=\NP$, there is no polynomial-time computable voting rule satisfying PJR, EJR, or FJR for all elections with $n=1$ and general utilities and general costs.
\end{remark}

Yet this reduction from Knapsack no longer works as-is if the utilities are guaranteed to be polynomial in the instance size.
This suggests that we may be able to avoid this computational barrier by restricting the utilities. In this section, we consider the computational complexity of finding PJR, EJR, and FJR committees in all PB instances $\elec = (N, C, B, \cost, (u_i)_{i \in N})$ such that $n \in O(1)$ \textit{and} additional constraints on $(u_i)_{i \in N}$ and $\cost$ are obeyed. 

Let $u_{max} := \max_{i \in N, c \in C} u_i(c)$ and $c_{max} := \max_{c \in C} \cost(c)$. We say that utilities are \emph{polynomial} if $u_{max}$ is bounded by a polynomial in $\abs{\elec}$. Since $\abs{\elec} = \poly{n, m, \log(c_{max}), \log(u_{max})}$\footnote{We assume that $B < \cost(C) \leq m \cdot c_{max}$ because otherwise all projects can trivially be selected. Hence $B$ is representable in $\log(m \cdot c_{max})$ space.}, this is equivalent to  $u_{max}$ being bounded by a polynomial in $n$, $m$, and $\log\paran{c_{max}}$. Analogously, we say that costs are polynomial if $c_{max}\in\poly{\abs{\elec}}$. For a problem to be in $\P$, it must be computable in time $\poly{|\elec|} = \poly{n, m, \log(c_{max}), \log(u_{max})}$.
Our results are summarized in \Cref{tab:constant-n-results}. 

\subsection{Greedy Cohesive Rule and Efficient Outcome Computation}

Indeed, we find that when the number of voters is constant and utilities are polynomial, the Greedy Cohesive Rule introduced in \cite{equalshares_fjr} and shown to always find an outcome satisfying FJR can be efficiently implemented.
\begin{definition}[Greedy Cohesive Rule (GCR) \cite{equalshares_fjr}]\label{def:gcr}
     Initially mark all voters as active and set the outcome ($W$) to empty. Find a tuple $(S, T, \beta)$ maximizing $\beta$ such that $S$ is a subset of the active voters, $T \subseteq C,$\footnote{%
    Differing from \cite{equalshares_fjr}, we consider $T \subseteq C$ instead of $T \subseteq C\setminus W$. We confirmed with the authors of \cite{equalshares_fjr} that this is necessary for GCR to satisfy FJR.
    } and $ \beta \in \mathbb{N}^+$ such that $|S|/n \geq \cost(T)/B$ and $\beta = \min_{i \in S}u_i(T)$. If no such tuple exists, return $W$. Otherwise, add $T$ to $W$, mark voters in $S$ as inactive, and repeat.
\end{definition}

The key idea is that if $2^n$ is not prohibitively large, e.g. when $n \in O(1)$, then if there is an efficient way of computing the $\beta$-maximizing $(S, T, \beta)$ for a fixed $S$, we can precompute the $\beta$-maximizing $(S, T, \beta)$ for \textit{every} fixed $S$. We can then sort the list of tuples by $\beta$ in $O(n2^n \cdot \log(m \cdot u_{max}))$ time, and execute GCR by moving through the list until a tuple is reached that corresponds to an $S$ for which all voters are currently active, an $O(n2^n)$ operation. 

\begin{restatable}{lemma}{gcrruntime}\label{lem:gcr-constant-n-runtime}
    GCR can be run in $O(n2^n \cdot \log(m \cdot u_{max}) \cdot m^{n+1} \cdot u_{max}^{n} \cdot \log(m \cdot c_{max}))$.
\end{restatable}

\begin{proof}
    The $\beta$-maximizing $(S, T, \beta)$ tuple for each $S \subseteq N$ is found using a DP algorithm. We provide detailed pseudocode and a proof of the DP correctness (\Cref{prop:dp-gcr-correct}) in \Cref{app:gcr-constant-n-runtime}. Fix some enumeration of both the voters in $S$ and the alternatives in $C$. The DP table has entries $A[j][x_1, \dots, x_{|S|}]$ for $j \in \{0, \dots, m\}$ and $x_1, \dots, x_{|S|} \in \{0, \dots, m \cdot u_{max}\}$, for a total of $(m+1) \times (m \cdot u_{max} + 1)^{|S|} \in O(m^{n+1} \cdot u_{max}^n)$ cells. The cell $A[j][x_1, \dots, x_{|S|}]$ stores the minimum cost of a subset of the first $j$ alternatives $T \subseteq \{c_1, \dots, c_j\}$ that achieves utility exactly $x_i$ for the $i$\textsuperscript{th} voter in $S$ for all $i \in \{1, \dots, |S|\}$. Each cell is computed with a constant number of operations on costs of sets of alternatives, which can be done in $O(\log(m \cdot c_{max}))$ time. The corresponding set of alternatives for each cell can be stored in a companion table with matching entries. The final tuple is chosen by scanning over the $m$\textsuperscript{th} row of $A$ and picking the cell guaranteeing the highest minimum utility for members of $S$ such that the cell entry is within the budget ``allotted'' to $S$.\qed
\end{proof}

 As shown in \cite[Prop. 3]{equalshares_fjr}, GCR satisfies FJR. For $\elec$ with $u_{max} \in \text{poly}(n, m, \log(c_{max}) )$ and $n \in O(1)$, \Cref{lem:gcr-constant-n-runtime} implies that we can find an outcome for $\elec$ satisfying FJR in time $\poly{m, \log(c_{max})}$. Hence, the runtime is in $\poly{|\elec|}$.

\begin{theorem}\label{thm:fjr-constant-n}
    For all $\elec$ with $n \in O(1)$, polynomial utilities, and general costs, an outcome satisfying FJR can be found in polynomial time.
\end{theorem}


\Cref{thm:fjr-constant-n} tells us that limiting utilities to be polynomial instead of general allows us to sidestep the computational intractability of FJR, EJR, and PJR in the case of $n \in O(1)$. Does the parallel approach of limiting costs to be polynomial afford us the same opportunity for computational efficiency? We find that the answer is \textit{partially}. We first give a modification of GCR satisfying EJR (and thus PJR) that can be implemented in polynomial time in this setting, but that stops short of satisfying FJR. Then, we show in \Cref{subsec:fjr-hard-constant-n} that \textit{no} polynomial-time computable voting rule $\mathcal{R}$ can guarantee an FJR-satisfying outcome for all elections, even when restricting to unit costs and two voters.

\begin{definition}[modified Greedy Cohesive Rule (mGCR)]
    The modified Greedy Cohesive Rule is identical to GCR (\Cref{def:gcr}), but the definition of $\beta$ is changed from $\min_{i \in S} u_i(T) = \min_{i \in S} \sum_{c \in T} u_i(c)$ to $\sum_{c \in T} \min_{i \in S} u_i(c)$.
\end{definition}

 The change in the definition of $\beta$ corresponds to a strengthening of the cohesiveness requirement put on coalitions. As a result, mGCR enforces EJR compliance in the same way that GCR effectively ``hardcodes'' FJR compliance.

\begin{restatable}{lemma}{mgcrejrsat}\label{lem:mgcr-ejr-sat}
mGCR satisfies EJR.
\end{restatable}

We defer the proof of \Cref{lem:mgcr-ejr-sat} to \Cref{app:mgcr-ejr-sat}, because it closely hews to the proof that GCR satisfies FJR in \cite{equalshares_fjr}. 

\begin{restatable}{lemma}{mgcrruntime}\label{lem:mgcr-runtime}
    mGCR can be implemented in time $O(n2^n \cdot m^2 \cdot c_{max} \cdot \log(m \cdot u_{max}))$.
\end{restatable}

\begin{proof}
    By the same argument as in \Cref{thm:fjr-constant-n}, if we can efficiently precompute the $(S, T, \beta)$ that maximizes $\beta$ for every $S \subseteq N$,  then the entire algorithm can be implemented to run in polynomial time. We can compute all relevant tuples $(S, T, \beta)$ as follows: for a given $S$, we assign all $c \in C$ the value $v_S(c) := \min_{i \in S} u_i(c)$, which takes $O(n \cdot m \cdot \log(u_{max}))$ time in total. We can then run the standard DP for Knapsack: we want the set $T \subseteq C$ that maximizes $\sum_{c \in T} v_S(c)$ subject to a maximum budget of $\lfloor B \cdot |S|/n \rfloor$ (which is the budget deserved by coalition $S$). 
    This DP subroutine utilizes an $(m+1) \times (\lfloor B \cdot |S|/n \rfloor + 1)$ table, where the cell at index $(j, b)$ stores the maximum value of $\sum_{c \in T} v_S(c)$ over all $T \subseteq \{c_1, \dots, c_{j}\}$. As $B \leq m \cdot c_{max}$, there are $O(m^2 \cdot c_{max})$ cells in the table, and each cell is filled using a constant number of operations on values that are at most $m \cdot u_{max}$. Hence, for every $S$ the total run time is $O(n \cdot m \cdot \log(u_{max})) + O(m^2 \cdot c_{max} \cdot \log(m \cdot u_{max}))$, which is in $O(n \cdot m^2 \cdot c_{max} \cdot \log(m \cdot u_{max}))$. Repeating this for the $2^n$ choices of $S$ yields the stated claim. \qed
\end{proof}

If we have $\elec$ such that $n \in O(\log m)$ and $c_{max} \in \poly{n, m, \log(u_{max})}$, then by combining \Cref{lem:mgcr-ejr-sat} and \Cref{lem:mgcr-runtime}, we know that running mGCR produces an outcome satisfying EJR in time $\text{poly}(m, \log(m \cdot u_{max}))$ and hence polynomial in $|\elec|$.

\begin{theorem}\label{thm:mgcr-constant-n}
    An outcome satisfying EJR can be computed in polynomial time for $n \in O(\log(m))$, general utilities, and polynomial costs. 
\end{theorem}


There are two natural follow-up lines of questioning that we discuss in the Appendix.  
In \Cref{app:bounded-voter-types}, we extend \Cref{thm:fjr-constant-n} and \Cref{thm:mgcr-constant-n} to elections with a general number of voters but with few distinct utility functions --- bounding the number of voter \textit{types} instead of the number of voters. In \Cref{app:constant-budget}, we show that if we constrain the budget to be constant instead of constraining the number of voters to be constant, then there exists a polynomial-time voting rule that satisfies FJR. 

\subsection{FJR is Weakly NP-Hard}\label{subsec:fjr-hard-constant-n}

Our positive result in \Cref{thm:mgcr-constant-n} turns out to be tight: Unless $\textbf{P}=\textbf{NP}$, there is no polynomial-time voting rule that satisfies FJR for all PB instances with constant voters, general utilities, and polynomial costs. In fact, this holds under the even more severe instance restriction of unit-costs and $n=2$.

\begin{restatable}{theorem}{fjrweakhardness}\label{thm:fjr-weak-hardness}
    Unless $\P=\NP$, there exists no polynomial-time computable voting rule $\mathcal{R}$ satisfying FJR for all elections with $n=2$, general utilities, and unit costs.
\end{restatable}

Similarly to the \textbf{NP}-hardness proof of \Cref{rem:hardness-constant-n} for $n=1$, we reduce from the Knapsack problem. However, in the unit-cost setting, finding the optimal Knapsack for a single voter can be done in linear time. To circumvent this, we insert the Knapsack problem into the problem of finding a committee satisfying FJR by using the utilities of one voter as the values and relating the utilities of the other voter to the costs. 

In the PB instance we construct, each voter has $k/2$ \emph{selfish} alternatives that give no utility to the other voter. The items from the Knapsack problem correspond to alternatives that give the first voter utility corresponding to the value and the second voter utility corresponding to a large constant minus the cost. We show that if there exists a solution to the Knapsack instance, the two voters can jointly demand the corresponding alternatives so that the committee of the $k$ selfish alternatives violates FJR. At the same time, if no solution to the Knapsack instance exists, we show that the committee of the $k$ selfish alternatives is the unique committee satisfying FJR. Thus, if and only if the voting rule returns this selfish committee, no solution to the Knapsack instance exists. We defer the formal proof to \Cref{app:fjr-weak-hardness}.

Together, \Cref{thm:fjr-constant-n,thm:mgcr-constant-n,thm:fjr-weak-hardness} in combination with \Cref{prop:implications} uncover a complete picture of the computational complexity for finding outcomes satisfying PJR, EJR, and FJR  for a constant number of voters.

\section{FJR Hardness}\label{sec:fjr-hardness}

While it is known that \emph{verifying} whether a given outcome satisfies PJR/EJR/FJR for an instance $\elec$ is \textbf{coNP}-complete even for committee elections with approval preferences \cite{aziz2018complexity,aziz-jr-ejr,kalayci2025full}, only very little is known about the complexity of \emph{finding some} outcome satisfying PJR/EJR/FJR. In particular, to the best of our knowledge, the only existing hardness result in this realm is that for general utilities and general costs, even for $n=1$ finding a PJR committee solves Knapsack (\Cref{rem:hardness-constant-n}). 
In this section, we give a hardness result for FJR for committee elections with additive utilities from a fixed range that does not rely on the size of the numbers in the instance, thus showing strong \textbf{NP}-hardness. 

\begin{theorem}\label{thm:fjr-strong-hardness}
    Unless $\P=\NP$, there exists no polynomial-time computable voting rule satisfying FJR, even for committee election instances where $u_i\colon C\to \set{0,1,...,5}$ for all $i\in N$.
\end{theorem}

We prove this via a reduction from a problem that we call \textsc{PartialTripleCover}. We show that partial triple cover is NP-hard in \Cref{app:partial-triple-cover}.

\alg{\textsc{PartialTripleCover}: Given a bipartite graph $G=(L\cup R, E)$, do there exist non-empty $S\subseteq L$ and $T\subseteq R$ such that $\abs{S}=\abs{T}$ and $\abs{\set{r \in T: \set{\ell,r} \in E}}\geq 3$ for all $\ell \in S$?}

The proof strategy is quite similar to the proof outlined in \Cref{subsec:fjr-hard-constant-n}: Here,  each voter has a single \emph{selfish} alternative giving only them utility; we use $k=n$ so that each voter can demand their selfish alternative. Furthermore, we insert the bipartite graph of a \textsc{PartialTripleCover} in the instance, so that a partial triple cover $S,T$ corresponds to a FJR violation of the \emph{selfish} committee, consisting of the $n$ selfish alternatives. Crucially, we also show that if no partial triple cover exists, the selfish committee is the only committee satisfying FJR. Thus, if and only if the voting rule returns the selfish committee, no partial triple cover exists.

\begin{proof}[\Cref{thm:fjr-strong-hardness}]
    Let us assume that we have a polynomial time voting rule that satisfies FJR on all committee election instances where $u_i:C\to\set{0,1,...,5}$ for all $i\in N$. We will show that we can use this voting rule to devise a polynomial time algorithm for \textsc{PartialTripleCover}, thus obtaining $\P=\NP$ by \Cref{lem:partial_triple_cover-hardness}.
    
    Given an instance $G=(L\cup R,E)$ of \textsc{PartialTripleCover}, we construct an election instance $\elec=(N,C,k,\paran{u_i}_{i \in N})$. For each vertex $\ell \in L$ we create one voter $i_\ell \in N$. For each vertex $r\in R$, we create one alternative $c_r\in C$, where voters gain utility $2$ from $c_r$ if $\ell,r$ share an edge in $G$, else $0$. 
    Furthermore, we create one \emph{selfish} alternative $d_\ell \in C$ for every voter $i_\ell \in N$, where voter $i_\ell$ gets utility $5$ from $d_\ell$ and all other voters get no utility. 
    We set $k=n=\abs{L}$.
    
    Consider the \emph{selfish} committee $W_0=\set{d_\ell}_{\ell \in L}\subseteq C$, which has size $n=k$. We know that $u_{i}(W_0)=5$ for all $i\in N$. We will show that if $G$ has a partial triple cover, $W_0$ does not satisfy FJR, but that if $G$ has no partial triple cover, $W_0$ is the unique committee that satisfies FJR. This implies a polynomial-time algorithm for \textsc{PartialTripleCover}: Let $W\subseteq C$, $\abs{W}\leq k$ be any committee satisfying FJR as returned by our voting rule; if and only if $W\neq W_0$ we know that $G$ has a partial triple cover.
    
    For the first direction, let $S',T'$ be a partial triple cover of $G$. Consider $S = \set{i_\ell: \ell \in S'}$ and $T=\set{c_r: r\in T'}$. We get that 
    $$u_{i_\ell}(T) = \sum_{c_r\in T}u_{i_\ell}(c_r) = 2\cdot \abs{r\in T': \set{\ell, r}\in E}\geq 6.$$ Thus, there exist $S\subseteq N$, $T\subseteq C$, such that $\abs{S}/n\geq \abs{T}/k$ and 
    $$\min_{i\in S} u_i(T) \geq 6 > 5=\max_{i\in S}u_i(W_0),$$ so $W_0$ does not satisfy FJR.
    
    For the other direction, assume $G$ has no partial triple cover. Let $W\subseteq C$ be any committee of size at most $k$ different from $W_0$. Let $S'$ be the set of all $\ell\in L$ such that $d_\ell \notin W$ and let $T'$ be the set of all $r\in R$ such that $c_r\in W$. Since $\abs{L\setminus S'} + \abs{T'} \leq k$ and $\abs{L\setminus S'} = k - \abs{S'}$, we know that $\abs{S'}\geq \abs{T'}$. Let $S''$ be an arbitrary subset of $S'$ of size $\abs{T'}$.
     If $\abs{T'}\geq 1$, we know that since $G$ doesn't have a partial triple cover, there exists a vertex $\ell^* \in S''$ such that $\abs{r\in T':\set{\ell^*, r}\in E} \leq 2$. If $\abs{T'}=0$, we know that such a $\ell^* \in S'$ exists, since $T'=\emptyset$ and $\abs{S'}\geq 1$.
    For this $\ell^*$, it holds that $$u_{i_{\ell^*}}(W) = \sum_{\ell \in L\setminus S'}u_{i_{\ell^*}}(d_\ell) + \sum_{r \in T'}u_{i_{\ell^*}}(c_r) =  2\cdot \abs{r\in T':\set{\ell^*, r}\in E} \leq 4.$$ However, for $S=\set{i_{\ell^*}}\subseteq N$, $T=\set{d_{\ell^*}}\subseteq C$, it holds that $\abs{S}/n\geq \abs{T}/k$ (both are $1/n$) and 
    $$\min_{i\in S}u_i(T) =5>4 \geq \max_{i\in S}u_i(W),$$ so $W$ doesn't satisfy FJR. Since there always exists a committee of size $k$ satisfying FJR (\Cref{prop:fjr-existance}), we get that $W_0$ is the unique committee satisfying FJR. \qed
\end{proof}

By a very similar reduction from \textsc{PartialTripleCover}, we also get a hardness result for cost-utilities from a bounded range. The proof is deferred to \Cref{app:fjr-hardness}.  

\begin{restatable}{theorem}{fjrstronghardnesscostutils}\label{thm:fjr-strong-hardness-cost-utils}
    Unless $\P=\NP$, there exists no polynomial-time computable voting rule satisfying FJR, even for cost-utilities with $\cost\colon C\to [5]$.
\end{restatable}

It is natural to wonder whether a similar proof technique can be used to show that unless $\P=\NP$, no polynomial-time voting rule satisfies EJR for committee elections. The main obstacle is that due to the stronger cohesiveness condition, the existence of a partial triple cover (in the graph constructed in the proof of \Cref{thm:fjr-strong-hardness}) no longer necessarily constitutes a violation of the axiom by the selfish committee; for this, the partial triple cover needs to be a biclique. In return, however, the absence of a balanced biclique does not suffice to guarantee that the selfish committee is the only committee satisfying the axiom. To alleviate this, we turn to the corresponding promise problem:

\alg{\textsc{BalancedBiclique-Or-No-Partial-Multi-Cover (Bb-Pmc)}: Given a bipartite graph $G=(L\cup R, E)$ and $q\in [\abs{L}]$ such that either $G$ has a size $q$ balanced biclique ($S\subseteq L$ and $T\subseteq R$ such that $\abs{S}=\abs{T}=q$ and $\set{\ell,r} \in E$ for all $\ell\in S$ and $r\in T$) or $G$ has no partial $q$-cover (a non-empty $S\subseteq L$ and $T\subseteq R$ such that $\abs{S}=\abs{T}$ and $\abs{\set{r \in T: \set{\ell,r} \in E}}\geq q$ for all $\ell \in S$), decide which of the two is the case.}

\begin{restatable}{proposition}{promisetoejr}\label{prop:promise-to-ejr}
    If there exists a polynomial-time computable voting rule that satisfies EJR for all committee elections where $u_i\colon C\to \set{0,1, ..., 2n}$, then \textsc{Bb-Pmc} is in \textbf{P}.
\end{restatable}

The proof of \Cref{prop:promise-to-ejr} is deferred to \Cref{sec:promise-to-ejr-proof}. It is unknown whether \textsc{Bb-Pmc} is \textbf{NP}-hard. Should this be the case, then by \Cref{prop:promise-to-ejr}, we get that unless $\P=\NP$, no polynomial-time voting rule satisfies EJR for committee elections with utilities bounded above by $2n$.

\section{PJR for Committee Elections}\label{sec:general-n-pjr}

In this section, we extend the Expanding Approvals Rule (EAR) from committee elections with ordinal preferences, due to Aziz and Lee \cite{aziz2020expanding}, to committee elections with additive utilities. EAR is known to satisfy rankPJR, a version of PJR for ordinal preferences. We show that EAR implies PJR in the more general setting of additive utilities. We start by extending its definition:

\begin{definition}
    The \emph{Expanding Approvals Rule (EAR) for additive utilities} starts with $W=\emptyset$, threshold $t= \max_{i\in N, c\in C}u_i(c)$ and assigns each voter a budget of $\nicefrac{k}{n}$. Repeatedly, it looks for an alternative $c\in C\setminus W$ for which the total budget of voters with $u_i(c)\geq t$ is at least $1$. If such a $c$ exists, it adds $c$ to $W$ and reduces the total budget of voters with $u_i(c)\geq t$ by $1$ (without making any voter's budget negative). If no such $c$ exists, it decreases $t$ to the next-smallest value in $\set{u_i(c): i\in N, c\in C}$, if $t$ already is the smallest value it returns $W$.
\end{definition} 

The proportionality guarantees of EAR that we prove hold for \emph{any} way of deducting a total budget of $1$ from the voters when adding an alternative to $c$,\footnote{Aziz and Lee \cite{aziz2020expanding} propose deducting the same fraction of their remaining budget from all voters, while in \cite{ejrplus,equalshares_fjr}, the authors consider deducting as-equal-as-possible absolute amounts from each voter. By \Cref{thm:ejr-impossible-for-sequential-committee}, we know that no choice of splitting the budget gives EJR, the next-stronger proportionality guarantee in our list.} and for any way of breaking ties when choosing $c$.

\begin{theorem}\label{thm:pjr-from-ear}
    For committee elections with additive utilities, the Expanding Approvals Rule satisfies PJR.
\end{theorem}

\begin{proof}
    Consider any $S \subseteq N, T \subseteq C$ with $\frac{|S|}{n} \geq \frac{|T|}{k}$. For $c\in T$, let $\alpha(c)=\min_{i\in S} u_i(c)$ be the minimum utility for $c$ in $S$. Sort $T=\set{c_1,...,c_\ell}$ in order of decreasing minimum utility in $S$, so that $\alpha(c_1)\geq ... \geq \alpha(c_\ell)$. We will show that for every $j\in [\ell]$, there exist at least $j$ alternatives $c \in W$ for which there exists $i\in S$ with $u_i(c)\geq \alpha(c_j)$. 
    
    For any $j\in [\ell]$, consider the moment when the algorithm is about to lower the threshold to the first $t < \alpha(c_j)$ (or, if $t$ already is minimal, is about to return $W$). Assume towards a contradiction that at this point there are at most $j-1$ alternatives $c \in W$ for which some $i\in S$ has $u_i(c)\geq \alpha(c_j)$. Thus, there exists $c^*\in\set{c_1,...,c_j}$ that is not yet in $W$. So far, voters have only spent budget on alternatives that they have utility at least $\alpha(c_j)$ for, so the total remaining budget of $S$ is at least $\nicefrac{k}{n}\abs{S}-(j-1)\geq \ell-j+1\geq 1$. Since $u_i(c^*)\geq \alpha(c_j)$ for all $i\in S$ and the total budget of $S$ is at least $1$, the algorithm would add $c^*$ before lowering the threshold, a contradiction.

    For $j=1,...,\ell$, let $w_j \notin \set{w_1,...,w_{j-1}}$ be an alternative in $W$ with $u_i(w_j)\geq \alpha(c_j)$; by the above argument we know that such $w_1,...,w_\ell$ exist. Thus, \begin{align*}
        \sum_{c \in W} \max_{i\in S} u_i(c) \geq \sum_{j=1}^\ell \max_{i\in S} u_i(w_j) \geq \sum_{j=1}^\ell \alpha(c_j) = \sum_{c \in T} \min_{i \in S} u_i(c). \tag*{\qed}
    \end{align*}
\end{proof}

\section{The Limitations of Sequential Voting Rules}\label{sec:greedy-limitations}

With one exception that we discuss later in this section, all known voting rules that are computable in polynomial time and satisfy justified representation axioms are \emph{greedy} in nature: They repeatedly pick a single alternative without ``planning ahead''. In this section, we formalize this using the notion of a \emph{sequential voting rule} from \cite{brill2023sequential} and extend it to allow for a constant amount of ``planning ahead''. We show that no voting rule of this type can satisfy EJR for committee elections and PJR for participatory budgeting. 

\begin{definition}[Sequential Voting Rules with $\beta$-lookahead] Given an election instance $\elec$, a sequence $W'=(c_1,...,c_\ell)$ of distinct alternatives in $C$, an alternative $c\in \set{\perp} \cup C\setminus W'$, and a lookahead set $L\subseteq C\setminus (W'\cup\set{c})$ a \emph{marginal contribution function} $f$ assigns a score $f(c,W',\elec\mid_{W'\cup\set{c}\cup L})\in \R$, where $\elec\mid_{W'\cup\set{c}\cup L}$ is $\elec$ restricted to only contain alternatives $C=W'\cup\set{c}\cup L$.

A voting rule $\mathcal R$ is \emph{sequential with $\beta$-lookahead} if for some fixed marginal contribution function $f$ its outcome 
$W=\set{c_1,...,c_\ell}$
is obtained by selecting $$c_j = \argmax_{c\in \set{\perp}\cup C\setminus W'_{j-1}}\max_{\substack{L\subseteq C \setminus(W'_{j-1}\cup\set{c})\\ \abs{L}\leq \beta }} f(c,W'_{j-1}, \elec\mid_{W'_{j-1}\cup\set{c}\cup L}),$$ breaking ties consistently\footnote{In particular, given is an order over all alternatives that is instance-independent, i.e. only depends on $(\cost(c),(u_i(c))_{i\in N})$, and is used only for tie-breaking. If two alternatives have the same $(\cost(c),(u_i(c))_{i\in N})$ they are identical, so no tie-breaking is needed. The results in this section hold for \emph{any} instance-independent order.}, for all $j\in [\ell+1]$ and $c_{\ell+1}=\perp$, where $W'_{j-1} = (c_1,...,c_{j-1})$. If a voting rule is sequential with $0$-lookahead, we say it is \emph{sequential}.
\end{definition}

In other words, a sequential voting rule with $\beta$-lookahead decides on the next alternative it adds to $W$\emdash or whether to terminate by selecting $\perp$\emdash solely based on a score of $c$, which only depends on $c$ itself, the alternatives in $W$ so far, and the lookahead set of at most $\beta$ additional alternatives. This is a natural extension of sequential voting rules: GCR satisfies FJR by looking at sets of alternatives $T$, trying to find multiple alternatives that complement each other well so that all $i\in S$ have high utility for $T$. We may hope (unfortunately, wrongfully, though) that looking for sets of up to $\beta$ alternatives, then adding them one-by-one, suffices.

\begin{theorem}\label{thm:ejr-impossible-for-sequential-committee}
For all $\beta \in \N$, no sequential voting rule with $\beta$-lookahead can give EJR, even for $n=\beta+2$, unit-costs, and utilities in $\set{0,...,2\beta+3}$.
\end{theorem}

\begin{proof}
    Let $N=[\beta+2]$.  For every voter $i \in N$, let $c_i$ be an alternative such that $u_i(c_i)=2\beta+3$ and $u_{i'}(c_i)=0$ for $i'\in N \setminus\set{i}$. Let $d$ be an alternative with $u_i(d)=2$ for every $i \in N$. Consider two committee election instances with voters $N$ and committee size $k=\beta+2$: $\elec_1$ with all alternatives $c_i$ for $i\in N$ and $\beta+1$ copies of $d$, and  $\elec_2$ with the same alternatives but an additional, $(\beta+2)$nd copy of $d$. 
    
    In $\elec_1$, the only committee satisfying EJR is $W_1=\set{c_1, ..., c_{\beta+2}}$, since if any $c_i$ is not in $W$, voter $i$ only gets at most utility $2\beta + 2$ from $W$ but can proportionally demand $c_i$, so utility $2\beta+3$. In $\elec_2$, the only committee satisfying EJR is $W_2$, the $\beta+2$ copies of $d$: First, note that $W_1$ does not satisfy EJR since every voter gets utility $2\beta +3$, but $S=N$ can demand $T=W_2$, with $\sum_{d \in W_2}\min_{i\in S} u_i(d)= 2(\beta + 2) > 2\beta + 3 = \max_{i \in S}u_i(W_1)$. In any $W$ that is not $W_1$, there exists $i\in N$ such that $c_i$ is not in $W$. If $W$ is also not $W_2$, we know that at most $\beta+1$ copies of $d$ are in $W$, so this voter has $u_i(W) \leq 2\beta + 2$; an EJR violation since $i$ can proportionally demand $c_i$, so utility $2\beta+3$. 
    
    Since $\elec_2$ is obtained from $\elec_1$ by adding an additional clone of an alternative which already exists $\beta+1$ times in $\elec_1$ (alternative $d$), we get that $\elec_1$ and $\elec_2$ are indistinguishable for a $\beta$-sequential voting rule when adding the first candidate to $W$: The set of values of the marginal contribution function are the same. By consistent tie-breaking, any voting rule needs to select the same first alternative in both $\elec_1$ and $\elec_2$. However, since the unique committees satisfying EJR in these elections, $W_1$ and $W_2$, are disjoint, we get that any $\beta$-sequential voting rule has to violate EJR in at least one of the two elections. \qed
    \end{proof}

\begin{theorem}\label{thm:pjr-impossible-for-sequential-cost-util}
No sequential voting rule with $\beta$-lookahead can give PJR, even for $n=1$ and cost-utilities in $\set{0,...,\beta+3}$.
\end{theorem}

\begin{proof}
    Let $N=\set{i}$. Let $c$ be an alternative with $\cost(c)=u_i(c)=\beta + 3$, and let $d$ be an alternative with $\cost(d)=u_i(d)=\beta + 2$. Set $B=(\beta+2)(\beta+3)$ and consider the following two instances: In $\elec_1$, there are $\beta + 2$ copies of $c$ and $d$ each. In $\elec_2$, there are $\beta+1$ copies of $c$ and $\beta + 3$ copies of $d$. 

    For cost utilities and a single voter, an outcome $W$ satisfies PJR if and only if there is no set $T\subseteq C$ with $\cost(W)<\cost(T)\leq B$. Using only alternatives of cost $\beta+2$ and $\beta+3$, the only way an outcome $W$ can have $\cost(W)=B$ is if it either consists of $\beta+2$ alternatives of cost $\beta+3$ or of $\beta+3$ alternatives of cost $\beta+2$. Thus, in $\elec_1$ the only outcome satisfying PJR is to take the $\beta+2$ copies of $c$, while in $\elec_2$ the only outcome satisfying PJR is to take the $\beta+3$ copies of $d$. 

    Since $\elec_1$ and $\elec_2$ only differ by having additional copies of alternatives that exist at least $\beta+1$ times in both, they are indistinguishable for a sequential voting rule with $\beta$-lookahead when adding the first candidate to $W$. However, since the unique outcomes satisfying PJR in these elections, $W_1$ and $W_2$, are disjoint, we get that any $\beta$-sequential voting rule has to violate PJR in at least one of the two elections. \qed
\end{proof}

Almost all rules that we are aware of that have proportionality guarantees and are polynomial time computable are sequential.
For budget-based rules like EAR, MES, or Phragm\'en's Sequential Rule \cite{brill2024phragmen}, $W'$ is sufficient to calculate the remaining budgets of the voters, from which we can calculate the score of $c$. For example, for EAR this would be the largest threshold $t$ at which $c$ can be selected. For other voting rules, like sequential Thiele methods \cite{Thiele1895}, Greedy Monroe \cite{monroe1995fully}, and the Maximin Support Method \cite{sanchez2024maximin}, it is evident from their definition that they are sequential.

The only non-sequential polynomial time voting rule satisfying a justified representation axiom that we are aware of is a local-search variant of PAV \cite{aziz2017polynomialtimealgorithmachieveextended}. However, PAV and local-search PAV are only defined for approval preferences with unit costs and the natural extensions of PAV to general utilities or general costs are known not to satisfy justified representation axioms.\footnote{For general costs, see \cite{equalshares_fjr}. For general utilities, consider $n=k=2$, $m=3$, $u_1(c_1)=100, u_1(c_2)=500$, $u_2(c_1)=1,u_2(c_3)=2$, all other utilities $0$: Any PAV-style rule would choose $W=\set{c_1,c_2}$, violating PJR for $S=\set{2}$.} Since local-search PAV aims to approximate the PAV outcome, it does not satisfy justified representation axioms in the more general setting for the same reasons. It is an interesting open question whether a different polynomial-time, local search algorithm can overcome the barriers for sequential algorithms.

Thus, any voting rule that breaks either of these boundaries needs to be structurally quite different from any polynomial-time voting rule considered in the literature so far.

\section{Discussion}

In this work, we move substantially closer to fully understanding the computational boundaries of achieving justified representation in participatory budgeting and committee elections with additive utilities. 

The result that finding an outcome satisfying FJR is strongly \textbf{NP}-hard and the limitations for sequential rules bring us closer to resolving two open questions posed in literature: Rey et al. \cite[Sec. 5.1.1.2]{rey2025computationalsocialchoiceindivisible} ask whether there exists a pseudo-polynomial time algorithm for EJR with general utilities and general costs; now we know that for FJR, the answer is no, and for EJR, no sequential rule can succeed. Lackner and Skowron \cite[Q5]{lackner_overview} ask whether an outcome satisfying FJR can be found in polynomial time for {approval} utilities and unit costs; now we know that with utilities in $\set{0,1,2,3,4,5}$, the answer is no. We believe it is an interesting open question whether a sequential rule can satisfy FJR for approval utilities and unit costs.

While not satisfying a justified representation axiom may disqualify a voting rule, satisfying a justified representation axiom alone is not enough to guarantee that a rule will be desirable. For example, GCR is designed specifically to satisfy FJR but does not meet other desirable criteria, such as laminar proportionality \cite{equalshares_fjr}. Whether FJR can be achieved by a more natural voting rule, or whether there exist practical voting rules giving stronger proportionality guarantees than those by MES and EAR are interesting problems; this paper shows that many natural rules, namely sequential rules, are unable to do the latter.  We believe that EAR, well-known in committee elections with ordinal preferences and clustering, is a natural and practical rule, though more work on its behavior with respect to other axioms is necessary.

\begin{credits}
\subsubsection{\ackname} This paper evolved from a class project in the Fall 2025 offering of Ariel Procaccia's class \emph{Optimized Democracy}. We thank Ariel Procaccia and Bailey Flanigan for helpful discussions and suggestions.
\end{credits}
%
\bibliographystyle{splncs04}
\bibliography{bibliography}
\appendix
\crefalias{section}{appendix}
\renewcommand{\theHsection}{A\arabic{section}}
\section{Appendix for \Cref{sec:constant-n}}\label[appendix]{app:constant-n}

We include pseudocode for the Greedy Cohesive Rule below in \Cref{alg:gcr} for reference.

\begin{algorithm}[H]
\caption{Greedy Cohesive Rule (GCR) \cite{equalshares_fjr}}
\label{alg:gcr}
\begin{algorithmic}[1]
\State $W \gets \emptyset$ \Comment{Selected outcome}
\State $V \gets N$ \Comment{Active voters}

\While{true}
    \State Find $S \subseteq V, T \subseteq C, \beta \in \mathbb{N}_{> 0}$ such that $|S|/n \geq \cost(T)/B$ and $\beta = \min_{i \in S}\sum_{c \in T}  u_i(c)$
    \If{no such $\beta, S, T$ exist}
        \State \Return $W$
    \EndIf

    \State Choose $(S, T, \beta)$ maximizing $\beta$
    \State $W \gets W \cup T$
    \State $V \gets V \setminus S$
\EndWhile
\end{algorithmic}
\end{algorithm}

\subsection{Proof of Efficient GCR DP Subroutine for \Cref{lem:gcr-constant-n-runtime}}\label[appendix]{app:gcr-constant-n-runtime}

In \Cref{alg:fjr_dp}, we provide pseudocode for the DP routine discussed in \Cref{lem:gcr-constant-n-runtime} that precomputes the $(S, T, \beta)$ for GCR. We note that $T$ can be constructed via a companion table that stores the choice of items corresponding to each cell in the primary DP table, $A$. Then in \Cref{prop:dp-gcr-correct}, we show correctness of this DP procedure. For ease of notation (so we can refer to prefixes of the alternative set), enumerate the alternative set $C = \{c_1, \dots, c_m\}$.

\begin{algorithm}
\caption{DP for GCR $(S, T, \beta)$ Precomputation}
\label{alg:fjr_dp}
\begin{algorithmic}[1]
\For{each subset $S \subseteq N$, relabel the voters as $= \set{1, \dots, \ell}$ and}
    \State \textbf{Initialize:} $A[0][0,\dots,0] \gets 0$
    \State \textbf{Initialize:} $A[0][x_1,\dots,x_\ell] \gets \infty$ for all $(x_1,\dots,x_\ell) \neq (0,\dots,0)$
    
    \For{$j = 1, \dots, m$}
        \For{each $x_1, \dots, x_\ell \in \set{0, \dots, m \cdot u_{max}}$}
            \State $v_{\text{excl}} \gets A[j-1][x_1, \dots, x_\ell]$
            
            \If{$x_i - u_i(c_j) \geq 0$ for all $i \in \set{1, \dots, \ell}$}
                \State $v_{\text{incl}} \gets \text{cost}(c_j) + A[j-1][x_1-u_1(c_j), \dots, x_\ell-u_\ell(c_j)]$
            \Else
                \State $v_{\text{incl}} \gets \infty$ \Comment{Candidate sets cannot induce negative utility}
            \EndIf
            
            \State $A[j][x_1, \dots, x_\ell] \gets \min \set{v_{\text{excl}}, v_{\text{incl}}}$
        \EndFor
    \EndFor
    
    \State $\beta \gets \max \set{x \in \mathbb{N} : \exists x_1, \dots, x_\ell \geq x \text{ s.t. } A[m][x_1, \dots, x_\ell] \leq \frac{B \cdot \ell}{n}}$
\EndFor
\end{algorithmic}
\end{algorithm}

\begin{proposition}\label{prop:dp-gcr-correct}
    For every $S =\{1, \dots, \ell\} \subseteq N$, the DP table $A$ constructed in \Cref{alg:fjr_dp} is such that for all $j \in \{0, \dots m\}$ and $x_1, \dots, x_{|S|} \in \{0, \dots, m \cdot u_{max}\}$ \[
    A[j][x_1, \dots, x_{\ell}] = \min_{T \in \mathcal{T}_j \text{ s.t. } u_i(T) = x_i \forall i \in S} \cost(T)
    \]
    where $\mathcal{T}_j = 2^{\{c_1, \dots, c_j\}}$ is the set of all possible $T$ that only contain alternatives from the first $j$. 
\end{proposition}

\begin{proof}
    Fix some $S = \{1, \dots, \ell\}$. The claim can be seen straightforwardly via induction. The base cases on lines 2 and 3 are accurate because when $j = 0$, $\mathcal{T}_j = \{\varnothing\}$. Therefore, the minimum over all $T \in \mathcal{T}_j$ of $T$ satisfying the utility conditions is $\cost(\varnothing)$ if the utilities are all 0, and infinity otherwise (because there is no $T \in \{\varnothing\}$ that can provide strictly positive utility for any voter).

    Now fix some $j \geq 1$ and assume that the induction hypothesis holds for all $k < j$. Fix any $x_1, \dots, x_{\ell} \in \{0, \dots, u_{max} \cdot m\}$, and refer to $u_i(T) = x_i$ for all $i \in S$ as the ``utility conditions.'' If there does not exist any $T \in \mathcal{T}_j$ that meets the utility conditions, then we know that there cannot exist a $T' \in \mathcal{T}_{j-1}$ such that \textit{either} (a) $u_i(T') = x_i - u_i(c_j)$ for all $i \in S$ \textit{or} (b) $u_i(T') = x_i$ for all $i \in S$. In (a), we could just set $T = T' \cup \{c_j\}$ and contradict the non-existence of a $T \in \mathcal{T}_j$ such that $u_i(T) = x_i$ for all $i \in S$. In the case of (b), we could just set $T = T'$ and arrive at the same contradiction. Hence, by the induction hypothesis, we have that $A[j-1][x_1 - u_1(c_j), \dots, x_{\ell} - u_{\ell}(c_j)] = A[j-1][x_1, \dots, x_{\ell}] = \infty$ so both $v_{incl}$ and $v_{excl}$ are set to infinity, as is $A[j][x_1, \dots, x_{\ell}]$ (which is as desired). 
    
    If there do exist sets of alternatives that satisfy the utility conditions, then fix some such $T \in \mathcal{T}_j$ with minimum cost among these sets. If $v_{excl} < \cost(T)$, then by the induction hypothesis, this would imply the existence of a set $T'$ meeting the utility conditions, and contradict the minimality of the cost of $T$ among such sets. Therefore $v_{excl} \geq \cost(T)$. Now consider $v_{incl}$. If $v_{incl} < \cost(T)$, then by the induction hypothesis there would have to exist a set $T'$ such that $T' \cup \{c_j\}$ meets the utility conditions, and $\cost(T') < \cost(T) - \cost(c_j) \implies \cost(T' \cup \{c_j\}) < \cost(T)$. This would also contradict the minimality of the cost of $T$. 
    
    If $c_j \in T$, then by the induction hypothesis, $A[j-1][x_1 - u_1(c_j), \dots, x_{\ell}-u_{\ell}(c_j)] = \cost(T) - \cost(c_j)$.  
    Therefore, $v_{incl} = \cost(c_j) + \cost(T) - \cost(c_j) = \cost(T)$. If $c_j \not \in T$, then we know that $T \in \mathcal{T}_{j-1}$, so $v_{excl} = \cost(T)$. Hence, either way, $A[j][x_1, \dots, x_{\ell}] \leftarrow \cost(T)$ on line 12, as desired. \qed
\end{proof}

\subsection{Proof of \Cref{lem:mgcr-ejr-sat}}\label[appendix]{app:mgcr-ejr-sat}

We include pseudocode for the modified Greedy Cohesive Rule below in \Cref{alg:mgcr} for reference.

\begin{algorithm}[H]
\caption{modified Greedy Cohesive Rule (mGCR) \cite{equalshares_fjr}}
\label{alg:mgcr}
\begin{algorithmic}[1]
\State $W \gets \emptyset$ \Comment{Selected outcome}
\State $V \gets N$ \Comment{Active voters}

\While{true}
    \State Find $S \subseteq V, T \subseteq C, \beta \in \mathbb{N}_{> 0}$ such that $|S|/n \geq \cost(T)/B$ and $\beta = \sum_{c \in T}  \min_{i \in S} u_i(c)$
    \If{no such $\beta, S, T$ exist}
        \State \Return $W$
    \EndIf

    \State Choose $(S, T, \beta)$ maximizing $\beta$
    \State $W \gets W \cup T$
    \State $V \gets V \setminus S$
\EndWhile
\end{algorithmic}
\end{algorithm}

\mgcrejrsat*

\begin{proof}
    First, we argue that any outcome returned by mGCR satisfies EJR. Fix an instance $\elec = (N, C, B, \cost, (u_i)_{i \in N})$.
    Let $W$ be the set returned by mGCR, and assume for the sake of contradiction that there exists an EJR violation: a set $S \subseteq N$ and $T \subseteq C$ such that $|S|/n \geq \cost(T)/B$ but $\max_{i \in S} u_i(W) < \sum_{c \in T} \min_{i \in S} u_i(c) := \beta$. 
    Let $d \in S$ be the first voter to have been deactivated in the run of mGCR, and note that the EJR violation implies that $u_d(W) < \beta$. 
    Note that such a voter must exist because mGCR would not terminate if all voters in $S$ were marked as active \,---\, it could mark them inactive by adding $T$ to $W$. 
    Let $\beta', S', T'$ be the values chosen on line 8 of \Cref{alg:mgcr} in the loop iteration that deactivated $d$. 
    As this is the loop where the first voter in $S$ is being deactivated, all voters in $S$ are active beforehand.
    Therefore, $\beta' \geq \beta$ because $(\beta, S, T)$ was a viable option, but $(\beta', S',T')$ was chosen instead. 
    By the condition on $\beta'$ from line 4, and the facts that $d \in S'$ and $T' \subseteq W$, we have that \[
    \beta' \leq u_d(T') \leq u_d(W) < \beta,
    \]
    a contradiction. Therefore there cannot be such an EJR violation.

    We also claim that mGCR will output an outcome that respects the budget constraint, so $\cost(W) \leq B$. 
    To see this, note that by the condition relating the size of $S$ to the cost of $T$, we have that whenever we add $T$ to $W$, we are deactivating at least $n \cdot \cost(T)/B$ voters. 
    Denote the set $T$ added at round $j$ as $T_j$. Then we have that $\cost(W) = \cost(\cup_j T_j) \leq \sum_j \cost(T_j)$ and $n \geq \sum_j n \cdot \cost(T_j)/B$ because the number of voters deactivated can't exceed $n$. 
    Rearranging yields that $B \geq \sum_j \cost(T_j) \geq \cost(W)$. \qed
\end{proof}

\subsection{Proof of \Cref{thm:fjr-weak-hardness}}\label[appendix]{app:fjr-weak-hardness}

\fjrweakhardness*

We prove this theorem by showing that a polynomial-time computable voting rule that returns an outcome satisfying FJR for any election with only $n=2$ voters can be used to devise a polynomial-time algorithm for a restricted version of the Knapsack problem. We prove that this restricted version of Knapsack is still $\NP$-hard by reducing from the original Knapsack problem. We start by formally defining \textsc{Restricted-Knapsack}.

\alg{\textsc{Restricted-Knapsack:} Given are $m$ items $((v_j,w_j))_{j\in[m]} \in (\N^2)^m$ and $s,v\in \N^+$ such that
\begin{align*}
    (v_j,w_j)=(v,v)  &\quad \forall  j\in[s],\\
    (v_j,w_j)=(0,0) &\quad \forall  j\in\set{s+1,...,2s},\\
    v_j\leq v, w_j\leq v &\quad \forall j\in [m].
\end{align*} Furthermore, we are guaranteed that for any $K'\subseteq [m]$ of size $\abs{K'}=2s$ such that 
$[2s]\neq K'$, it holds that $\sum_{j\in K'}v_j\neq vs$ and 
$\sum_{j\in K'}w_j\neq vs$. Does there exist $K\subseteq [m]$ of size $\abs{K}=2s$ such that $\sum_{j\in K} v_j \geq vs+1$ and 
$\sum_{j\in K} w_i \leq vs-1$?
} 

\begin{lemma}\label{lem:restricted-Knapsack-Hardness}
    \textsc{Restricted-Knapsack} is $\NP$-hard.
\end{lemma}

    We reduce from the standard Knapsack problem, known to be $\NP$-hard \cite{Karp1972}. We assume that all values and weights are positive integers; it is straightforward to check that this assumption makes the problem no easier. \alg{\textsc{Knapsack}: Given are $\hat{m}$ items $((\hat{v}_i,\hat{w}_i))_{i\in[\hat{m}]} \in ({\N^+}^2)^{\hat{m}}$ and $\hat{T},\hat{B} \in \N^+$. Does there exist $\hat{K}\subseteq [\hat{m}]$ such that $\sum_{i\in \hat{K}} \hat{v}_i \geq \hat{T}$ and $\sum_{i\in \hat{K}} \hat{w}_i \leq\hat{B}$?}  

  \begin{proof}[\Cref{lem:restricted-Knapsack-Hardness}]  Given an instance of \textsc{Knapsack}, we create the following instance of \textsc{Restricted-Knapsack}: Let $m=3\hat{m}+1$, $s=\hat{m}$, and $v=2 \paran{\hat{T} + \hat{B} + \sum_{i\in [\hat{m}]} \hat{v}_i+ \sum_{i\in [\hat{m}]} \hat{w}_i + 1}$. The first $s$ items with $(v_i,w_i)=(v,v)$ for $i \in [s]$ are called the `heavy' items, and the next $s$ items with $(v_i,w_i)=(0,0)$ for $i \in \{s+1, \dots, 2s\}$ are called the `filler' items. For every $i \in [\hat{m}]$ and the corresponding element in the original \textsc{Knapsack} instance, $(\hat{v}_i,\hat{w}_i)$, add an element $({v}_{i+2s},{w}_{i+2s})=(2\hat{v}_i,2\hat{w}_i)$ to the new \textsc{Restricted-Knapsack} instance. Finally, add $({v}_{3s+1},{w}_{3s+1})=(v-2\hat{T}+1,v-2\hat{B}-1)$, and call it the `threshold' item.

    First, let us verify that this is a valid instance of \textsc{Restricted-Knapsack}. It holds that $0\leq v_i,w_i\leq v$ for all $i$. Furthermore, if a set $K'\subseteq [m]$ of size $2s$ includes all $s$ heavy items and any item with $i>2s$, we know $\sum_{i\in K'}v_i,\sum_{i\in K'}w_i>vs$. If it includes at most $s-2$ of the $s$ heavy items or exactly $s-1$ of the $s$ heavy items but not the threshold item, we know that $\sum_{i\in K'}v_i,\sum_{i\in K'}w_i<vs$. Lastly, if it includes exactly $s-1$ of the $s$ heavy items and the threshold item, we know that $\sum_{i\in K'}v_i,\sum_{i\in K'}w_i \equiv 1 \not\equiv 0 \equiv vs \mod 2$ so $\sum_{i\in K'}v_i,\sum_{i\in K'}w_i \neq vs$. We can thus conclude that for every $K'\subseteq [m]$ of size $k$, $K'\neq [2s]$, it holds that $\sum_{i\in K'}v_i \neq vs$ and $\sum_{i\in K'}w_i \neq vs$.

    Let us now show that the original \textsc{Knapsack} instance is a YES instance if and only if the \textsc{Restricted-Knapsack} instance is a YES instance. 

    ``$\Rightarrow$'': Let $\hat{K}$ be a solution to the \textsc{Knapsack} instance. We construct the following candidate solution set $K$ for the \textsc{Restricted-Knapsack} instance: the set of elements corresponding to $\hat{K}$, i.e. $i+2s\in K$ if $i\in \hat{K}$, together with $s-1$ heavy items, the threshold item, and $s-\lvert\hat{K}\rvert\in [s]$ filler items. We get a set $K\subseteq [m]$ of $\lvert\hat{K}\rvert + (s-1) + 1 + (s- \lvert\hat{K}\rvert) = 2s$ items such that $$\sum_{i\in K} v_i = vs - 2\hat{T}+1 + \sum_{i\in \hat{K}}2\hat{v}_i \geq vs + 1,$$ $$\sum_{i\in K} w_i = vs - 2\hat{B}-1 + \sum_{i\in \hat{K}}2\hat{w}_i \leq vs - 1.$$

    ``$\Leftarrow$'': Let $K$ be a solution to the \textsc{Restricted-Knapsack} instance. We let $\hat{K}$ be the set of all non-heavy, non-filler, non-threshold items in $K$, that is for any $i\in[\hat{m}]$, it holds that $i\in \hat{K}$ if and only if $i+2s \in K$. As argued above, we know that $\sum_{i\in K} v_i \geq vs + 1$ and $\sum_{i\in K} w_i \leq vs - 1$ is only possible if $K$ includes $s-1$ heavy items and the threshold item. Thus, we get that $$vs - 2\hat{T}+1 + \sum_{i\in \hat{K}} 2\hat{v}_i = \sum_{i\in K} v_i \geq vs + 1,$$ $$vs - 2\hat{B}-1 + \sum_{i\in \hat{K}}  2\hat{w}_i = \sum_{i\in K} w_i \leq vs - 1.$$
    This implies that $\sum_{i\in \hat{K}} \hat{v}_i\geq \hat{T}$ and $\sum_{i\in \hat{K}} \hat{w}_i\leq \hat{B}$, a solution to the knapsack problem.\qed
\end{proof}

With \Cref{lem:restricted-Knapsack-Hardness} in hand, we can proceed to prove \Cref{thm:fjr-weak-hardness}.

\begin{proof}[\Cref{thm:fjr-weak-hardness}]
Let us assume that we have a polynomial time voting rule that satisfies FJR on all committee election instances with $n=2$. We will show that we can use this voting rule to devise a polynomial time algorithm for \textsc{Restricted-Knapsack}, thus obtaining $\P=\NP$ by \Cref{lem:restricted-Knapsack-Hardness}.

Given an instance of \textsc{Restricted-Knapsack}, we create the following committee election instance. Let $N=\set{1,2}$. For all $j\in [m]$, create an alternative $c_j$ such that $u_1(c_j)= v_j \in \set{0,...,v}$ and $u_2(c_j)=v-w_j\in\set{0,...,v}$. In our language from the main part of the paper, $c_j$ for $j\in[s]$ and $j\in[s+1,2s]$ are the selfish alternatives for voters $1$ and $2$, respectively. We set $k=2s$.

We let $W_0=\set{c_j}_{j \in [2s]}\subseteq C$ be the committee of size $k$ (the selfish committee) corresponding to the $k$ ``seeded'' items in the \textsc{Restricted-Knapsack} instance. We know that $u_1(W_0)=u_2(W_0)=vs$. We will show that if the \textsc{Restricted-Knapsack} instance is satisfiable, $W_0$ does not satisfy FJR; but that if the \textsc{Restricted-Knapsack} instance is unsatisfiable, $W_0$ is the unique committee that satisfies FJR. This implies a polynomial-time algorithm for \textsc{Restricted-Knapsack}: Let $W\subseteq C$, $\abs{W}=k$ be any committee satisfying FJR as returned by our polynomial-time voting rule;
we know that the \textsc{Restricted-Knapsack} instance is satisfiable if and only if $W\neq W_0$.

Let us first assume that the \textsc{Restricted-Knapsack} instance is a YES instance, that is, there exists $K\subseteq [m]$ of size $2s=k$ such that $\sum_{j\in K} v_j\geq vs + 1$ and $\sum_{j\in K} w_j\leq vs-1$. Consider $S=\set{1,2}$ and $T=\set{c_j}_{j\in K}$. We get that $\abs{S}/n\geq \abs{T}/k$ (as both are $1$) and $$u_1(T) = \sum_{c_j\in T} u_1(c_j) = \sum_{j\in K} v_j \geq vs + 1,$$ 
$$u_2(T) = \sum_{c_j\in T} u_2(c_j) = \sum_{j\in K} (v - w_j) \geq vk - (vs - 1)=vs+1.$$ 
Thus, $\min_{i\in S} u_i(T) \geq vs+1 > vs = \max_{i\in S} u_i(W_0)$, so $W_0$ does not satisfy FJR. 

Now, assume that the \textsc{Restricted-Knapsack} instance is unsatisfiable. Let $W\subseteq C$ be a committee of size $k$, other than $W_0$. Consider $K=\set{j: c_j \in W}$. The \textsc{Restricted-Knapsack} instance being unsatisfiable implies that $\sum_{j\in K} v_j < vs+1$ or $\sum_{j\in K} w_j > vs-1$. We know that $\sum_{j\in K} v_j\neq vs$ and $\sum_{j\in K} w_j\neq vs$, so it needs to hold that either $\sum_{j\in K} v_j \leq vs-1$ or $\sum_{j\in K} w_j \geq vs+1$. We get that either 
$$u_1(W)=\sum_{c_j\in W} u_1(c_j) = \sum_{j\in K} v_j \leq vs - 1$$ 
or 
$$u_2(W)=\sum_{c_j\in W} u_2(c_j) = \sum_{j\in K} (v-w_j) \leq vk-(vs +1) = vs-1.$$ 
Let $i^*$ be said voter for which $u_{i^*}(W)\leq vs-1$ and consider $S=\set{i^*}$ and $T=\set{c_{s(i^*-1)+1},...,c_{si^*}}$. We get that $\abs{S}/n\geq \abs{T}/k$ (as both are $1/2$) and $\min_{i\in S} u_i(T) = vs > vs-1 \geq \max_{i\in S} u_i(W)$, so $W$ doesn't satisfy FJR. Since adding alternatives to a committee cannot introduce new FJR violations, this implies that no committee $W\subseteq C$ of size \emph{at most} $k$, other than $W_0$, can satisfy FJR. There always exists a committee of size $k$ satisfying FJR \cite{equalshares_fjr}, so we get that $W_0$ is the unique committee satisfying FJR. \qed
\end{proof}

\subsection{Extension to Elections with Bounded Voter Types}\label[appendix]{app:bounded-voter-types}

\Cref{thm:fjr-constant-n} and \Cref{thm:mgcr-constant-n} extend to the setting with an arbitrary number of voters, as long as the number of distinct utility functions, which we call the \textit{number of voter types}, is bounded. Crucially, the definition of PJR, EJR and FJR does not depend on the number of voters of a given type, but only on the different types of voters present in a given coalition.

\begin{theorem}\label{thm:fjr-bounded-voter-types}[Extension of \Cref{thm:fjr-constant-n}]
    For all $\elec$ with a constant number of voter types, polynomial utilities, and general costs, an outcome satisfying FJR can be found in polynomial time.
\end{theorem}

\begin{proof}
    Since the definition of FJR only depends on the set of voter types and the total number of voters present in a given coalition $S \subseteq N$ (but not on the number of voters of each type), we can similarly implement GCR in polynomial time. It suffices to modify the dynamic program of \cref{alg:fjr_dp} in the following ways: for a given set $S \subseteq N$, we enumerate the $\ell$ distinct types of voters present in it, and create a DP table with entries $A[j][x_1, \dots, x_\ell]$ for $j \in \{0, \dots, m\}$ and $x_1, \dots, x_\ell \in \{0, \dots, m \cdot u_{max} \}$, which store the minimum cost of a subset of the first $j$ alternatives $T \subseteq \{c_1, \dots, c_j\}$ that achieves utility exactly $x_i$ for the $i$\textsuperscript{th} voter type in $S$ for all $i \in \{1, \dots, \ell\}$. The analysis of the correctness and runtime of this modified dynamic program for the precomputation of the $\beta$-maximizing tuple $(S,T,\beta)$ for a given $S \subseteq N$ is completely analogous to \Cref{prop:dp-gcr-correct}. The corresponding set of alternatives for each cell can similarly be stored in a companion table with matching entries. Finally, it remains to adapt the implementation of the greedy cohesive rule in the following ways: instead of going through the $\beta$-maximizing tuples $(S,T,\beta)$ for all coalitions $S$, we only go through all subsets of active voter types and all feasible sizes of coalitions, of which there are $O(2^{t}n)$ (where $t$ is the number of voter types), and choose the one which maximizes $\beta$ as precomputed by the above DP. Then we can choose an arbitrary subset of voters that realizes this set of voter types. Crucially, the bound on $t$ means that only $\poly{|\elec|}$ tuples need to be computed.  A voter type is deactivated once all voters of this type are deactivated. The analysis of the correctness and runtime of this adapted GCR algorithm is otherwise completely analogous to \Cref{thm:fjr-constant-n}. \qed
\end{proof}

\begin{theorem}\label{thm:mgcr-bounded-voter-types}[Extension of \Cref{thm:mgcr-constant-n}]
     For all $\elec$ with a number of distinct voter types $t \in O(\log m)$, general utilities, and polynomial costs, an outcome satisfying EJR can be found in polynomial time. 
\end{theorem}

\begin{proof}
     Since the definition of EJR and the mGCR objective only depend on the set of voter types and the total number of voters present in a given coalition $S \subseteq N$ (but not on the number of voters of each type), we can similarly implement mGCR in polynomial time in this setting. It again suffices to consider the subsets $S'$ of voter types as well as the possible sizes $\ell$ of coalitions, and analogously define $v_{S'}(c)$ to be the minimum utility for alternative $c$ by any voter type in $S'$. As in \Cref{thm:mgcr-constant-n}, we can use the standard Knapsack DP to find a $\beta$-maximizing tuple $(S',T,\beta)$, and adapt the execution of mGCR by iterating through the $O(2^{t}n) = O(m^{O(1)} n)$ possible subsets of active voter types and feasible coalition sizes. \qed
\end{proof}

\subsection{Elections with Constant Budget}\label[appendix]{app:constant-budget}

\begin{theorem}\label{cor:fjr-k}
    Let $B$ be a fixed integer. Given any participatory budgeting instance with budget $B$, we can find an outcome satisfying FJR in time $O(m^{2B} n^{2})$.
\end{theorem}

To prove this, we show that verifying whether a committee satisfies FJR can be done efficiently for constant $B$:

\begin{lemma}\label{thm:fjr-k-verification}
   Let $B$ be a fixed integer. Given any participatory budgeting instance with budget $B$ and an outcome $W$ of cost at most $B$, we can verify in time $O(m^{B}n^2)$ whether $W$ satisfies FJR.
\end{lemma}

\begin{proof}
    $W$ doesn't satisfy FJR if and only if there exist $S\subseteq N$ and $T\subseteq C$ satisfying $\frac{|S|}{n} \geq \frac{\cost(T)}{B}$ and with $\max_{i \in S} u_i(W) < \min_{i \in S}  u_i(T)$. 
    
    Since each project has cost at least one, a set $T \subseteq C$ of cost at most $B$ contains at most $B$ alternatives.
    We iterate over all $T \subseteq C$ of size between 1 and $B$, discarding those of cost greater than $B$, to check if for any $T$, a set $S\subseteq N$ meeting the condition above exists. 

    For a fixed $T$, we iterate over the voters in $N$. For each $i\in N$, we let $S=\{j\in N: u_j(W) \leq u_i(W) < u_j(T)\}$. If $\abs{S}\geq \frac{\cost(T)}{B} n$, we stop and output that FJR is violated. If we iterated over all $i\in N$ for all $T$, without stopping, we return that FJR is satisfied.
    
    For correctness, first note that the algorithm only outputs that FJR is violated when it found $S$ and $T$ that constitute this violation. Thus, it only remains to show that the algorithm will always report that FJR is violated if $W$ doesn't satisfy FJR, i.e. if there exist $S\subseteq N$ and $T\subseteq C$ such that $\frac{\abs{S}}{n}\geq \frac{\cost(T)}{B}$ and $\max_{i \in S} u_i(W) < \min_{i \in S}  u_i(T)$. We let $i^S_{\max}(W)= \arg\max_{i\in S} u_i(W)$ be the voter in $S$ with the highest utility for $W$. 
    We know that $S\subseteq\set{j\in N: u_j(W)\leq u_{i^S_{\max}}(W) < u_j(T)}$, so for this $T$ and voter $i^S_{\max}(W)\in N$, the algorithm above found a set $S'\supseteq S$. Since $\abs{S'}\geq\abs{S}\geq \frac{\cost(T)}{B} n$, the algorithm returned that there exists an FJR violation.

    Note that there are $O(m+...+m^B)=O(m^B)$ such $T$ of size between 1 and $B$.
    For each $T$, we can do the search for an $S$ meeting the conditions in $O(n^2)$ time. This gives a total runtime of $O(m^B n^2)$. \qed
\end{proof}

\begin{proof}[\Cref{cor:fjr-k}]
    Since there exist only $O(m^B)$ outcomes of size at most $B$, \Cref{thm:fjr-k-verification} immediately implies a $O(m^B\cdot m^B \cdot n^2 )$ time voting rule for FJR: Iterate over all outcomes and check for each if it satisfies FJR; return the first that does.
\end{proof}

\section{Partial Triple Cover is \textbf{NP}-hard}\label[appendix]{app:partial-triple-cover}
Let us recall the definition of partial triple cover.

\alg{\textsc{PartialTripleCover} : Given a bipartite graph $G=(L\cup R, E)$, do there exist non-empty $S\subseteq L$ and $T\subseteq R$ such that $\abs{S}=\abs{T}$ and $\abs{\set{r \in T: \set{\ell,r} \in E}}\geq 3$ for all $\ell \in S$?}

We show that partial triple cover is \textbf{NP}-hard by reduction from the 3-regular subgraph problem. We refer to \cite{regularSubgraph} for a proof that \textsc{$3$-regularSubgraph} is $\NP$-hard.\footnote{We note that in fact, the graph constructed in \cite{regularSubgraph} in their reduction from exact cover by $3$-sets (\textsc{X3C}) to \textsc{3-regularSubgraph} can be applied, without change, to show that \textsc{PartialTripleCover} is $\NP$-hard by reduction from \textsc{X3C}: The graph they construct has a partial triple cover if and only if the original instance had an exact cover by $3$-sets. However, since the graph construction is quite involved\emdash though all the more beautiful\emdash we give a shorter, independent reduction from \textsc{3-regularSubgraph} here.}

\alg{\textsc{3-regularSubgraph}: Given a bipartite graph $G=(L\cup R, E)$, does $G$ contain a non-empty subgraph that is $3$-regular (that is, every vertex has degree exactly 3)?}

\begin{lemma}\label{lem:partial_triple_cover-hardness}
    \textsc{PartialTripleCover} is $\NP$-hard.
\end{lemma}

\begin{proof}
    Let $G=(L\cup R, E)$ be an instance of \textsc{3-regularSubgraph}. We construct an instance $G'=(L'\cup R', E')$ of \textsc{PartialTripleCover} as follows: $L'$ contains a vertex $b_r$ for every vertex $r\in R$ and vertices $c_\ell$ and $d_\ell$ for every vertex $\ell\in L$. For every edge $\set{\ell, r}\in E$ (where $\ell\in L, r\in R$ in $G$),  $R'$ contains a vertex $a_{\ell,r}$. We connect the vertices $a_{\ell,r}\in R'$ to $b_r, c_\ell$, and $d_\ell$. That is, for any $a_{\ell,r}\in R'$, $E'$ contains $\set{a_{\ell,r}, b_r}, \set{a_{\ell,r}, c_\ell}, \set{a_{\ell,r}, d_\ell}$. This graph construction is illustrated in \Cref{fig:P3C-reduction}.

    \begin{figure}[]
\centering
\begin{tikzpicture}[
    scale=0.8,
    every node/.style={font=\small},
    every path/.style={line width=0.8pt}
  ]
  \tikzset{
    whitenode/.style={circle, draw, fill=gray!20, inner sep=1.5pt, minimum size=0.6cm},
    blacknode/.style={circle, draw, fill=black!80, text=white, inner sep=1.5pt, minimum size=0.6cm}
  }

  \node[whitenode] (v) at (-4,1) {$v$};
  \node[whitenode] (u) at (-4,-1) {$u$};

  \node[blacknode] (xL) at (-2,2) {$x$};
  \node[blacknode] (yL) at (-2,0) {$y$};
  \node[blacknode] (zL) at (-2,-2) {$z$};

  \node[whitenode] (cv) at (2,2.5) {$c_v$};
  \node[whitenode] (dv) at (2,1.5) {$d_v$};

  \node[whitenode] (cu) at (2,-0.5) {$c_u$};
  \node[whitenode] (du) at (2,-1.5) {$d_u$};

  \node[blacknode] (avx) at (4,2) {$a_{v,x}$};
  \node[blacknode] (avy) at (4,0.75) {$a_{v,y}$};
  \node[blacknode] (auy) at (4,-0.5) {$a_{u,y}$};
  \node[blacknode] (auz) at (4,-2) {$a_{u,z}$};

  \node[whitenode] (x) at (6,2) {$b_x$};
  \node[whitenode] (y) at (6,0) {$b_y$};
  \node[whitenode] (z) at (6,-2) {$b_z$};

  \draw (v) -- (xL);
  \draw (v) -- (yL);
  \draw (u) -- (yL);
  \draw (u) -- (zL);

  \draw (dv) -- (avx);
  \draw (cv) -- (avx);
  \draw (cv) -- (avy);
  \draw (dv) -- (avy);
  \draw (avx) -- (x);
  \draw (avy) -- (y);

  \draw (cu) -- (auy);
  \draw (du) -- (auy);
  \draw (cu) -- (auz);
  \draw (du) -- (auz);
  \draw (auy) -- (y);
  \draw (auz) -- (z);

\end{tikzpicture}
\caption{An example of the reduction of an instance of \textsc{3-regularSubgraph} on the left to an instance of \textsc{PartialTripleCover} on the right. The vertices in $L$ and $L'$ are light gray, while the vertices in $R$ and $R'$ are black.}
\label{fig:P3C-reduction}
\end{figure}
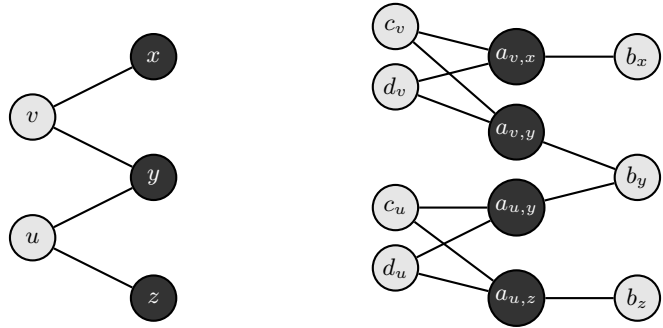

We now show that $G=(L\cup R, E)$ has a 3-regular subgraph if and only if $G'=(L'\cup R', E')$ has a partial triple cover.

    ``$\Rightarrow$'': Assume $G=(L\cup R, E)$ has a  3-regular subgraph $H=(L_H \cup R_H, E_H)$. Then select  $S=\set{b_r: r\in R_H}\cup \set{c_\ell, d_\ell: \ell\in L_H}\subseteq L'$ and $T=\set{a_{\ell,r}:\set{\ell,r} \in E_H, \ell\in L, r \in R}\subseteq R'$. Since $H$ is a $3$-regular bipartite graph, $\abs{E_H}=3\abs{L_H}=3\abs{R_H}$, thus $\abs{S}=\abs{T}$. Furthermore, for any $b_r \in S$, we have that $r\in R_H$, so there are exactly 3 vertices $\ell \in L_H$ so that $\set{\ell,r}\in E_H$. Thus, for any $b_r\in S$, $$
    \abs{\set{a_{\ell,r} \in T: \set{a_{\ell,r}, b_r}\in E'}} =  \abs{\set{\ell \in L_H: \set{\ell,r} \in E_H}} = 3.$$ We can argue analogously for any $c_\ell, d_\ell\in S$. Thus, there exists a partial triple cover of $G'$.
    
    ``$\Leftarrow$'': Assume $G'=(L'\cup R', E')$ has a partial triple cover $(S,T)$. Then pick $L_H=\set{\ell \in L: c_\ell \in S}$, $R_H=\set{r \in R: b_r\in S}$, and $E_H=\set{\set{\ell, r}\in E: \ell \in L, r\in R, a_{\ell,r} \in T}$. All $a_{\ell,r}\in R'$ have degree exactly $3$ in $G'$ and $b_r,c_\ell,d_\ell \in S$ have at least $3$ neighbors in $T\subseteq R'$. Thus, every $b_r,c_\ell,d_\ell \in S$ has exactly $3$ neighbors in $T$ and every $a_{\ell,r} \in T$ has exactly $3$ neighbors in $S$. We get that for any  $r\in R_H$, $$\abs{\set{\ell \in L: \set{\ell,r} \in E_H}}=\abs{\set{\ell \in L: a_{\ell,r} \in T}}=\abs{\set{a_{\ell,r} \in T: \set{a_{\ell,r}, b_r}\in E'}}=3$$ and for any  $\ell \in L_H$, $$\abs{\set{r \in R: \set{\ell,r} \in E_H}}=\abs{\set{r \in R: a_{\ell,r} \in T}}=\abs{\set{a_{\ell,r} \in T: \set{a_{\ell,r}, c_\ell}\in E'}}=3.$$ Thus, $H=(L_H \cup R_H, E_H)$ is a 3-regular subgraph of $G$. \qed
\end{proof}

\section{Proof of \Cref{thm:fjr-strong-hardness-cost-utils}}\label[appendix]{app:fjr-hardness}

\fjrstronghardnesscostutils*

\begin{proof} 
    Let us assume that we have a polynomial time voting rule that satisfies FJR on all PB instances with cost utilities and $\cost\colon C\to [5]$. We will show that we can use this voting rule to devise a polynomial time algorithm for \textsc{PartialTripleCover}, thus obtaining $\P=\NP$ by \Cref{lem:partial_triple_cover-hardness}.
    
    Given an instance $G=(L\cup R,E)$ of \textsc{PartialTripleCover}, we construct a PB instance $\elec=(N,C,B,\cost, \paran{u_i}_{i\in N})$: For each vertex $\ell \in L$ we create five voters $i_\ell^1,...,i_\ell^{5} \in N$. For each vertex $r\in R$, we create one alternative $c_r\in C$ with $\cost(c_r)=2$, where voters $i^1_\ell, i^2_\ell$ gain utility from $c_r$ if $\ell,r$ share an edge in $G$, all other voters do not gain utility from $c_r$.
    Furthermore, we create alternatives $d_\ell \in C, d'_\ell \in C$ with $\cost(d_\ell)=5$ and $\cost(d'_\ell)=3$ for every $\ell \in L$. Voters $i^1_\ell,...,i^5_\ell$ get utility from $d_\ell$, voters $i^3_\ell,...,i^5_\ell$ get utility from $d'_\ell$, and all other voters get no utility. 
    We set $B=n=5\abs{L}$.
    
    Let $W_0=\set{d_\ell}_{\ell \in L}\subseteq C$ be an outcome of total cost $5\abs{L}=B$. We know that $u_{i}(W_0)=5$ for all $i\in N$. We will show that if $G$ has a partial triple cover, $W_0$ does not satisfy FJR, but if $G$ has no partial triple cover, $W_0$ is the unique outcome that satisfies FJR. This implies a polynomial-time algorithm for \textsc{PartialTripleCover}: Let $W\subseteq C$ be any outcome satisfying FJR as returned by our voting rule; if and only if $W\neq W_0$ we know that $G$ has a partial triple cover.
    
    For the first direction, let $S',T'$ be a partial triple cover of $G$. Consider $S = \set{i^1_\ell, i^2_\ell: \ell \in S'}$ and $T=\set{c_r: r\in T'}$. We get that $$u_{i^1_\ell}(T) = \sum_{c_r\in T}u_{i^1_\ell}(c_r) = 2\cdot \abs{r\in T': \set{\ell, r}\in E}\geq 6$$ and analogously $u_{i^2_\ell}(T)\geq 6$. Thus, there exist $S\subseteq N$, $T\subseteq C$, such that $\frac{|S|}{n}\ge \frac{\cost(T)}{B}$ (since $\frac{\abs{S}}{n} = \frac{2\abs{T'}}{n} = \frac{\cost(T)}{B}$) and $\min_{i\in S} u_i(T) \geq 6>5=\max_{i\in S}u_i(W_0),$ so $W_0$ doesn't satisfy FJR.
    
For the other direction, assume $G$ has no partial triple cover. Let $W\subseteq C$ be any outcome of cost at most $B$ other than $W_0$. Let $S'$ be the set of all $\ell \in L$ such that $d_\ell \notin W$ and let $T'$ be the set of all $r \in R$ such that $c_r\in W$.
First, assume $\abs{S'}<\abs{T'}$. Then, there exists some $\ell \in S'$ so that $d'_\ell \notin W$, since otherwise $\cost(W)\geq 5\cdot \abs{L\setminus S'} + 2 \cdot \abs{T'} + 3 \cdot \abs{S'}> 5 \cdot \abs{L} =B$. 
Since also $d_\ell \notin W$ by definition, we get that $W$ violates FJR for $S=\set{i^3_\ell, ..., i^5_\ell}$ and $T=\set{d'_\ell}$, since $\min_{i\in S} u_i(T)=3>0=\max_{i\in S}u_i(W)$. 

Thus, let us now consider $\abs{S'}\geq \abs{T'}$.  Let $S''$ be an arbitrary subset of $S'$ of size $\abs{T'}$. If $\abs{T'}\geq 1$, we know that since $G$ doesn't have a partial triple cover, there exists a vertex $\ell^* \in S''$ such that $\abs{r\in T':\set{\ell^*, r}\in E} \leq 2$. If $\abs{T'}= 0$, the same holds for some $\ell^* \in S'$ since $T'=\emptyset$ and $\abs{S'}\geq 1$. For this $\ell^*$, we know that $$u_{i^1_{\ell^*}}(W) = \sum_{\ell \in L\setminus S'}u_{i^1_{\ell^*}}(d_\ell) + \sum_{r \in T'}u_{i^1_{\ell^*}}(c_r) =  2\cdot \abs{r\in T':\set{\ell^*, r}\in E} \leq 4$$ and analogously $u_{i^2_{\ell^*}}(W)\leq 4$. Additionally, since $d_{\ell^*}\notin W$, we know that $u_{i^3_{\ell^*}}(W)=u_{i^4_{\ell^*}}(W)=u_{i^5_{\ell^*}}(W)\leq 3$. However, for $S=\set{i^1_{\ell^*},...,i^5_{\ell^*}}\subseteq N$, $T=\set{d_{\ell^*}}\subseteq C$, it holds that $\abs{S}/n = \cost(T)/B$ (both are $5/n$) and $$\min_{i\in S} u_i(T) =5>4\geq \max_{i\in S}u_i(W),$$ so $W$ doesn't satisfy FJR. Since there always exists an outcome of cost at most $B$ satisfying FJR (\Cref{prop:fjr-existance}), we get that $W_0$ is the unique outcome satisfying FJR. \qed \end{proof}

\section{Proof of \Cref{prop:promise-to-ejr}}\label{sec:promise-to-ejr-proof}

\promisetoejr*

\begin{proof}[\Cref{prop:promise-to-ejr}]
    Let us assume that we have a polynomial-time voting rule that satisfies EJR on all committee election instances where $u_i:C\to\set{0,1,..., 2n}$ for all $i\in N$. We will show that we can use this voting rule to devise a polynomial-time algorithm for \textsc{Bb-Pmc}.
    
    Given an instance $G=(L\cup R,E)$ and $q$ of \textsc{Bb-Pmc}, we construct an election instance $\elec=(N,C,k,\paran{u_i}_{i \in N})$, using almost the same construction as in the proof of \Cref{thm:fjr-strong-hardness}: For each vertex $\ell \in L$ we create one voter $i_\ell \in N$. For each vertex $r\in R$, we create one alternative $c_r\in C$, where voters gain utility $2$ from $c_r$ if $\ell,r$ share an edge in $G$, else $0$. 
    Furthermore, we create one \emph{selfish} alternative $d_\ell \in C$ for every voter $i_\ell \in N$, where voter $i_\ell$ gets utility $2q-1$ from $d_\ell$ and all other voters get no utility. 
    We set $k=n=\abs{L}$.
    
    Consider the \emph{selfish} committee $W_0=\set{d_\ell}_{\ell \in L}\subseteq C$, which has size $n=k$. We know that $u_{i}(W_0)=2q-1$ for all $i\in N$. We will show that if $G$ has a size $q$ balanced biclique, $W_0$ does not satisfy EJR, but that if $G$ has no partial $q$-cover, $W_0$ is the unique committee that satisfies EJR. This implies a polynomial-time algorithm for \textsc{Bb-Pmc}: Let $W\subseteq C$, $\abs{W}\leq k$ be any committee satisfying EJR as returned by our voting rule; if and only if $W\neq W_0$ we know that $G$ has a size $q$ balanced biclique.
    
    For the first direction, let $S',T'$ be a size $q$ balanced biclique of $G$. Consider $S = \set{i_\ell: \ell \in S'}$ and $T=\set{c_r: r\in T'}$. 
    Since $\set{\ell, r}\in E$ for all $\ell \in S'$ and $r\in T'$, we get that  $u_{i}(c)=2$ for all $i\in S$ and $c\in T$. Thus, there exist $S\subseteq N$, $T\subseteq C$, such that $\abs{S}/n \geq \abs{T}/k$ and 
    $$\sum_{c\in T} \min_{i\in S} u_i(c) = \sum_{c\in T} 2 = 2q > 2q-1=\max_{i\in S}u_i(W_0),$$ so $W_0$ does not satisfy EJR.
    
    For the other direction, assume $G$ has no partial $q$-cover. Let $W\subseteq C$ be any committee of size at most $k$ different from $W_0$. Let $S'$ be the set of all $\ell\in L$ such that $d_\ell \notin W$ and let $T'$ be the set of all $r\in R$ such that $c_r\in W$. Since $\abs{L\setminus S'} + \abs{T'} \leq k$ and $\abs{L\setminus S'} = k - \abs{S'}$, we know that $\abs{S'}\geq \abs{T'}$. Let $S''$ be an arbitrary subset of $S'$ of size $\abs{T'}$.
     If $\abs{T'}\geq 1$, we know that since $G$ doesn't have a partial $q$-cover, there exists a vertex $\ell^* \in S''$ such that $\abs{r\in T':\set{\ell^*, r}\in E} \leq q-1$. If $\abs{T'}=0$, we know that such a $\ell^* \in S'$ exists, since $T'=\emptyset$ and $\abs{S'}\geq 1$.
    For this $\ell^*$, it holds that $$u_{i_{\ell^*}}(W) = \sum_{\ell \in L\setminus S'}u_{i_{\ell^*}}(d_\ell) + \sum_{r \in T'}u_{i_{\ell^*}}(c_r) =  2\cdot \abs{r\in T':\set{\ell^*, r}\in E} \leq 2(q-1).$$ However, for $S=\set{i_{\ell^*}}\subseteq N$, $T=\set{d_{\ell^*}}\subseteq C$, it holds that $\abs{S}/n\geq \abs{T}/k$ (both are $1/n$) and 
    $$ \sum_{c\in T}\min_{i\in S}u_i(c) = u_{i_{\ell^*}}(d_{\ell^*}) = 2q-1>2(q-1)\geq \max_{i\in S}u_i(W),$$ so $W$ doesn't satisfy EJR. Since there always exists a committee of size $k$ satisfying EJR (\Cref{prop:implications,prop:fjr-existance}), we get that $W_0$ is the unique committee satisfying EJR. \qed
\end{proof}
\end{document}